\documentclass[11pt]{article}

\usepackage[margin=1in]{geometry}

\usepackage{amsmath, amsfonts, amssymb}

\usepackage[ruled,vlined]{algorithm2e} 
\SetAlFnt{\small}
\SetAlCapFnt{\small}
\SetAlCapNameFnt{\small}
\SetAlCapHSkip{0pt}
\IncMargin{-\parindent}

\usepackage{float,placeins}
\DeclareMathOperator*{\argmax}{arg\,max}

\usepackage{amsthm}
\usepackage{thm-restate}
\usepackage{thmtools}

\usepackage{multirow}

\usepackage{cleveref}
\usepackage{indentfirst}

\usepackage{makecell}

\renewcommand\cite[1]{\citep{#1}}

\usepackage{xcolor}
\newcommand{\Comments}{1}
\newcommand{\mynote}[3]{\ifnum\Comments=1\textcolor{#1}{#2: #3}\fi}

\usepackage{censor}

\usepackage{color}
\usepackage{subcaption}

\usepackage[group-separator={,}]{siunitx}
\newcommand\blfootnote[1]{%
	\begingroup
	\renewcommand\thefootnote{}\footnote{#1}%
	\addtocounter{footnote}{-1}%
	\endgroup
}

\newcommand\CO{\ensuremath{CO}}
\newcommand\CQ{\ensuremath{CQ}}

\newtheorem{proposition}{Proposition}

\usepackage[sort&compress,numbers]{natbib}
\usepackage{url}
\usepackage{graphicx}

\usepackage{booktabs} 

\begin{document}

\title{Optimizing Library Usage and Browser Experience:\\ Application to the New York Public Library}

\author{
\begin{minipage}{0.45\textwidth}
    \centering
    Zhi Liu\\
    Cornell Tech\\
    \texttt{zl724@cornell.edu}
\end{minipage}
\hfill
\begin{minipage}{0.45\textwidth}
    \centering
    Wenchang Zhu\\
    Cornell Tech\\
    \texttt{wz368@cornell.edu}
\end{minipage}
\vspace{1cm}
\\
\begin{minipage}{0.45\textwidth}
    \centering
    Sarah Rankin\\
    The New York Public Library\\
    \texttt{sarahrankin@nypl.org}
\end{minipage}
\hfill
\begin{minipage}{0.45\textwidth}
    \centering
    Nikhil Garg\\
    Cornell Tech\\
    \texttt{ngarg@cornell.edu}
\end{minipage}
}

\date{}

\maketitle

\begin{abstract}
    We tackle the challenge brought to urban library systems by the {holds system} -- which allows users to request books available at other branches to be transferred for local pickup. The holds system increases usage of the entire collection, at the expense of an in-person browser's experience at the source branch. We study the optimization of usage and browser experience, where the library has two levers: where a book should come from when a hold request is placed, and how many book copies at each branch should be available through the holds system versus reserved for browsers. We first show that the problem of maximizing usage can be viewed through the lens of revenue management, for which near-optimal fulfillment policies exist. We then develop a simulation framework that further optimizes for browser experience, through book reservations. We empirically apply our methods to data from the New York Public Library to design implementable policies. We find that though a substantial trade-off exists between these two desiderata, a balanced policy can improve browser experience over the historical policy without significantly sacrificing usage. Because browser usage is more prevalent among branches in low-income areas, this policy further increases system-wide equity: notably, for branches in the 25\% lowest-income neighborhoods, it improves both usage and browser experience by about 15\%.
\end{abstract}

\blfootnote{
			We thank colleagues at the New York Public Library for their constructive feedback. We also benefited from discussions with Gabriel Agostini, Sidhika Balachandar, Erica Chiang, Evan Dong, Sophie Greenwood, Kenny Peng, and Divya Shanmugam. 
}

\section{Introduction}


Modern public library systems are a dynamic network consisting of multiple {branches}. Each branch has a unique collection (and high degrees of independence to curate it), engages with the surrounding community, and serves as a neighborhood learning hub. 
However, library branches are also {interconnected} and together compose a \textit{system} in which books can flow between branches. As a prominent example, the New York Public Library (NYPL) hosts
books among 92 branches across the boroughs of Manhattan, the Bronx, and Staten Island,\footnote{The boroughs of Queens and Brooklyn have separate systems.} serving over 800 thousand active \textit{patrons} (users) in 2023 \citep{nypl2023}. 

%

One way through which branches are connected is the \textit{holds} system: an online portal allowing patrons to request a book available at any branch in the system; a copy of the book is then transferred to their local branch to be picked up. Holds systems increase effective system-wide \textit{usage} of books and thus {efficiency}: books desired in one branch are not sitting unused in another. However, this increase in efficiency comes at a cost. A recent study by \citet{liu2024identifying} found one pitfall of the holds system: as shown in \Cref{fig:intro_disparities}, patrons in some neighborhoods use the holds system more frequently, and as a result, books from branches in other neighborhoods are depleted to satisfy these hold requests.  More generally, by removing book copies from shelves, the holds system worsens the \textit{browsing} experience---books that are not checked out serve patrons who browse the collection and read books inside the branch; serving these patrons cannot be overlooked as libraries are a public resource. Historically, the effect on bookshelf collection quality due to holds requests varied throughout the system. Furthermore, as usage of the holds system is higher in higher-income neighborhoods, branches in lower-income neighborhoods are disproportionately depleted, raising equity concerns.

In this paper, we ask: \textit{can we retain the efficiency benefits of the holds system while mitigating its effects on the browser experience}, by modifying library operations? We consider the optimization of two policy levers: (1) inventory reserve: What fraction of copies for each book in each branch's collection should be made available to the holds system versus \textit{reserved} for in-person browsers at that branch; and (2) hold requests fulfillment: When a hold request is placed, where should the copy of book come from, provided that it's available in multiple branches.

By studying these levers, we connect the problem of managing the holds system to a broader class of problems on resource allocation in a networked setting, in the revenue management literature. Books in a library branch can be considered as a flexible resource in a network, where it can be used to satisfy patrons both at its home branch, and those who use the holds system at other branches. However, one unique aspect of our problem that differs from other applications to which revenue management methods are applied is that books sitting on the shelf have direct value to browsers, unlike products sitting in a warehouse. Intuitively, then, a book copy should be pulled from a branch which has multiple copies of the book, and for which the demand from browsers is low. 
These levers are further implementable by the library system. Regarding inventory reserves: the NYPL has been experimenting with small collections reserved for browser usage in select branches, as a temporary measure to ensure browsers get access to some of the most popular books, and it is practically feasible to scale up such practice. Regarding hold requests fulfillment: we show that our Pareto-optimal policies can be implemented by translating them into parameters to a tier-based paging method that is currently in use, without substantial performance degradation.

\begin{figure}[tb]
    \centering
    \begin{subfigure}[b]{0.325\textwidth}
        \centering
        \includegraphics[width=\textwidth]{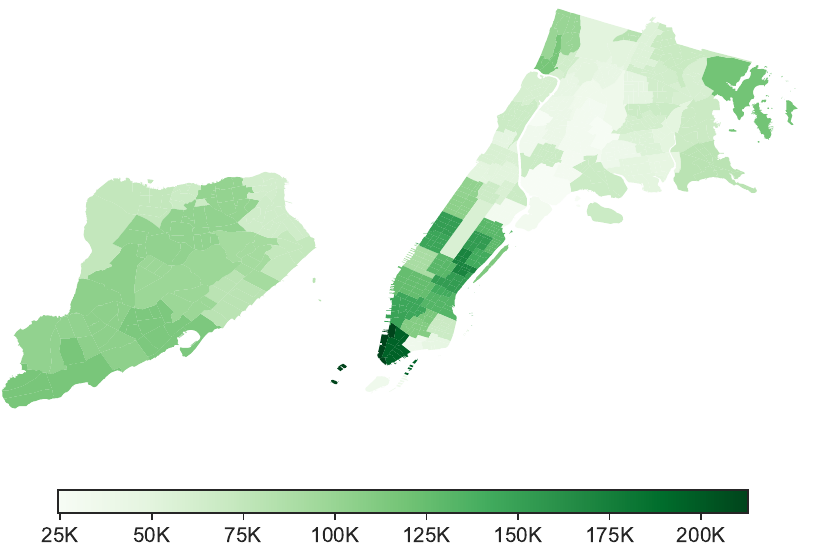}
        \caption{Median income}
        \label{fig:1a}
    \end{subfigure}
    \begin{subfigure}[b]{0.325\textwidth}
        \centering
        \includegraphics[width=\textwidth]{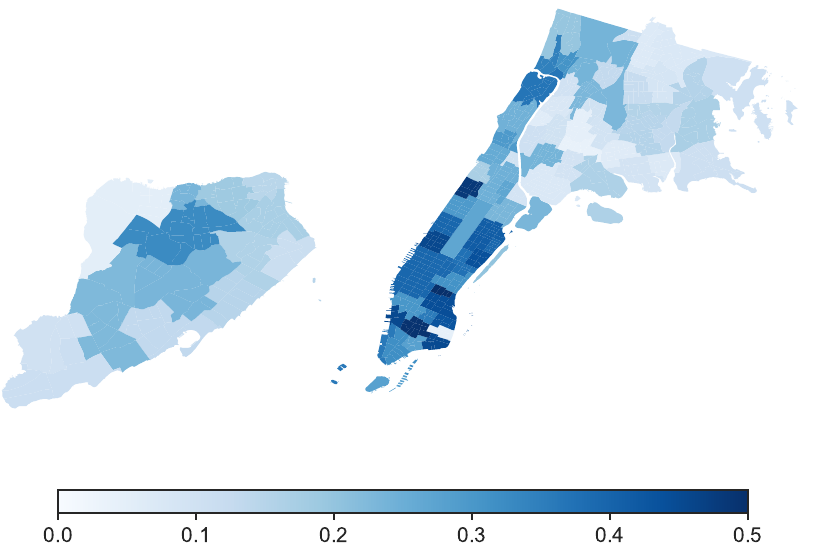}
        \caption{Hold usage fraction}
        \label{fig:1b}
    \end{subfigure}
    \begin{subfigure}[b]{0.325\textwidth}
        \centering
        \includegraphics[width=\textwidth]{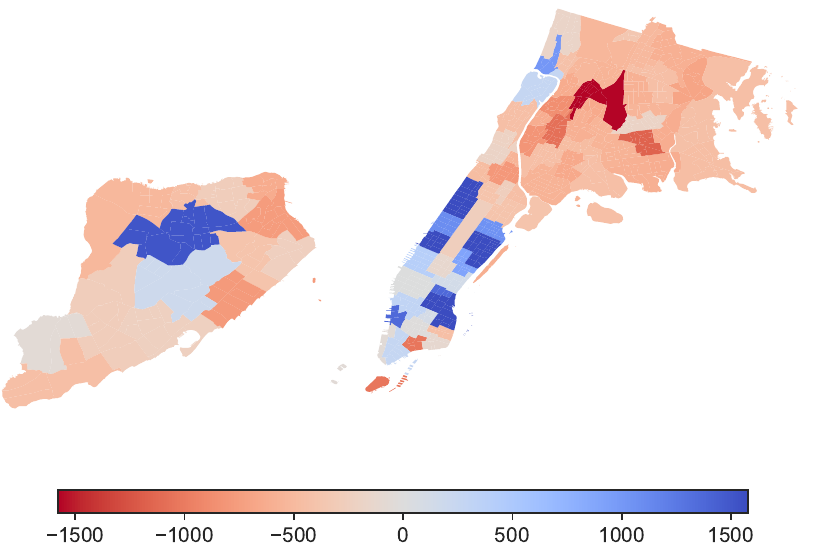}
        \caption{Net book desirability inflow}
        \label{fig:1c}
    \end{subfigure}
    \caption{(a) The median income of patrons served by each library branch differs substantially across the boroughs of Manhattan (center), the Bronx (upper right), and Staten Island (left). (b) Substantial differences in patrons' tendency to use the holds system among branches, correlated with neighborhood income levels, as shown in \citet{liu2024identifying}. For example, branches in Manhattan tend to serve higher-income patrons, who also use the holds system more, compared to branches in the Bronx. (c) Hold usage disproportionately pulls books from low-hold-usage and lower-income branches and thus worsens browser experience there. This effect is measured through the \textit{net desirability inflow} of each branch, calculated as the difference between the number of books taken from other branches to fulfill holds at this branch and the number of books taken from this branch to fulfill holds at all other branches, weighted by the desirability of each book title. Book desirability, as measured by \citet{liu2024identifying}, is a score between 0 and 1 assigned to each book title, with higher scores indicating a higher probability of getting checked out conditional on availability. We note that, once a book is returned by the holds user, it is returned to the source library. 
    }
    \label{fig:intro_disparities}
\end{figure}

We take the following steps towards more efficient and more equitable operations of the holds system: (a) we first formulate and analyze the problem of maximizing usage, and show that it can be cast through the lens of revenue management problems with flexible products, for which a near-optimal method for hold request fulfillment exists; (b) we then extend the model to capture browser experience, and provide a framework for finding implementable policies through a simulation-optimization approach; (c) we empirically derive optimal policies on the Pareto frontier of usage-experience trade-off using a simulator calibrated to data from the NYPL, and show that a more balanced policy possess desirable properties for practitioners.

\paragraph{A stylized model for maximizing usage} 
A book circulating within the library system can fulfill hold requests from {any} branch in the library system network, as well as serve local patrons. On the other hand, a patron's hold request may be fulfilled flexibly: either by book in the local branch, or any other branch in the library system network. We draw a connection to problems studied in the revenue management literature, where a resource that can be used to satisfy different product requests needs to be allocated to maximize some reward. We start by formulating this problem as a dynamic program of maximizing usage using a holds fulfillment method, under a set of initial inventory levels. We study a linear programming (LP) relaxation of this dynamic program, and show that there exists a $\frac12$-optimal fulfillment method for maximizing usage, based on a value function approximation of the LP. Such an approach is often used for operational challenges arising in e-commerce and airlines, but has not been applied to study the operation of public interest systems, to the best of our knowledge. 

\paragraph{Simulation optimization framework} However, taking into account browser experience within the revenue management approach seems intractable. Thus, we develop a simulation-optimization framework that finds a balance between usage and browser experience, by optimizing the inventory reserved for browsers at each branch. The simulation takes as inputs parameters estimated from and calibrated to historical checkout data, which are used to generate checkouts.

Then, we derive and evaluate different policies, defined as the combination of hold fulfillment method and browser reserves. First, we fix the hold fulfillment to the 1/2-optimal policy calculated above. We then search over inventory reserve fractions, using the simulator: we use simulation to measure each policy's impact on desired usage and browser experience metrics. Then we optimize for the inventory reserve of each branch, through a Bayesian optimization procedure. This Bayesian optimization approach explores the Pareto frontier between overall usage and browser experience objectives. Finally, we take policies found on the usage-browser experience Pareto frontier, and \textit{simplify} them into implementable policies, in which hold fulfillment is based on a static tier list of branches, shared across books, compatible with current practice. We then simulate the performance of these implementable counterparts. This simulation optimization framework allows us to measure fine-grained objective functions, and has broader applicability in other domains, where the diversity and quality of the inventory itself are important.


\paragraph{Empirical characterization of optimal policies} We empirically apply our approach to data from the New York Public Library to design policies that better balances usage and browser experience. We choose objectives based on the maximum achievable usage and browser experience. We find that (1) the trade-off between efficiency (overall usage) and browser experience is significant: the policy that maximizes browser experience trades 40\% of efficiency loss for 17\% of browser experience gain compared with the most efficient policy; (2) despite this trade-off, there exist policies that Pareto dominate historical practices: these policies induce more usage and better browser experience overall; (3) a more balanced policy also has more nuanced implications on individual branches, such as reducing disparities in net desirability inflow of books, and improving availability of most popular books---compared to the historical policy, this policy especially increases usage and the browser experience in branches in low-income neighborhoods, at little cost to usage in other branches. We have begun discussions about the feasibility of putting these policy innovations into practice, through our collaboration with the New York Public Library.

\subsection*{Related work}
Our work connects to the broad literature on public interest operations, algorithmic fairness and efficiency-equity tradeoffs, and network revenue management---across communities in economics, computer science, and operations. 

\paragraph{Public-interest operations and efficiency-equity tradeoffs in practice} 
Our work broadly connects to the large, growing literature on the equity-efficiency trade-off in public-interest organizations and allocation of public resources and goods \cite{monachou2022fairness}, both theoretically and in practice, particularly from an operational perspective: e.g., allocation of medical supplies \citep{manshadi2021fair}, food donations \citep{sinclair2022sequential, banerjee2023online}, social services \citep{koenecke2023popular,azizi2018designing}, education programs \citep{Meier2023Effectiveness}, transportation \citep{maheshwari2024congestion, torrico2024equitable, ostrovsky2024effective,jalota2021efficiency,delarue2024algorithmic,banerjee2019incorporating,bertsimas2020bus}, response to non-emergency incidents \citep{liu2024redesigning}, multi-sided recommendations and ranking \citep{chen2023interpolating,greenwood2024useritem,patroranking,manshadi2023redesigning}, and theoretical characterizations of efficiency-equity trade-offs more broadly \citep{bertsimas2011price,bertsimas2012efficiency,liang2022algorithmic}.  Within this literature, the most related is the work of \citet{liu2024redesigning} who, in the context of allocation policies for maintenance resources in New York City, also empirically explore an equity-efficiency Pareto frontier using a simulation-optimization approach. 

A common question of substantial interest throughout this literature is: \textit{when are tradeoffs large, both theoretically as a function of the structure of the problem, and empirically in practice in a concrete setting?}. To this literature, we contribute: (a) a novel application within which to study this broad question, in the study of library system operations, through hold fulfillment and browser reserve policies; (b) relatedly, a novel methodological domain (network revenue management with flexible products, as discussed below), from which policies are derived and optimized; (c) an empirical answer to this question in our real-world application, showing that the tradeoff can be severe but that we can find \textit{implementable} policies that achieve substantial Pareto improvements over the historical policy.




\paragraph{Library operations} The application of our work -- urban public library resource allocation -- relates to a growing literature that looks into the operations of library system through quantitative measures. Existing research usually focuses on quantifying high-level statistics, such as large-scale disparities in funding \citep{sin2011neighborhood} or accessibility \citep{cheng2021assessing, guo2017spatial} across an entire country or region, or comparison between rural and urban libraries \citep{real2014rural}. More fine-grained studies are also emerging: \citet{del2021urban} studies urban library system in Medellin, Colombia as a network using a data envelopment analysis (DEA), and estimate various impacts on efficiency; \citet{liu2024identifying} study the New York Public Library using a Bayesian latent variable model, and estimate key parameters contributing to efficient usage, with an emphasis on the holds system. Both papers uncover socio-economic inequities, with branch usage correlated or affected by the socio-economic status of the neighborhood. These works are essential for understanding the challenges modern libraries are facing but are not yet sufficient in providing guidance for practitioners to solve these challenges. We empirically use the parameter estimates of \citet{liu2024identifying}, and further calibrate a simulator to evaluate and optimize our policies. 

We directly tackle the challenge of balancing usage and browser experience in urban libraries. To this line of work, we contribute a computational framework for deriving feasible and implementable policies of the holds system, an important, widely used component of library operations. Our usage objective aligns with the DEA method used by many prior works in this field (e.g., \citet{shim2003applying, reichmann2004measuring, del2021urban}), where the efficiency of an entire system is evaluated using weighted aggregation of inputs and outputs of individual system units. Our browser experience objective is proposed through extensive conversation with library practitioners, and is novel to the field. 

\paragraph{Revenue management and simulation optimization} Library systems consist of a large network of branches that provide similar `products' (books) to `customers' (patrons) with different needs. This setting relates to the broader class of problems known as revenue management (see, e.g., \citet{talluri1998analysis,besbes2012blind, ferreira2018online, ma2020approximation}), which includes, for example, managing customer bookings in a network of flights. In the holds system, a hold request coming into a branch can be fulfilled by books at the local branch or at any other branch. Such property is referred to as `flexible' in the literature \citep{gallego2004revenue, gallego2004managing}. Our modeling results apply the technical tools developed by \citet{zhu2024performance}, which considers a network of resources with limited capacity, and requests for a flexible product may be satisfied by multiple combinations of resources. They provide a policy that achieves at least $1/(1+L)$ of the optimal expected revenue, where $L$ is the maximum amount of resources each product request consumes. In our setting, $L=1$ as each arriving patron requests at most one copy of the book, and, consistent with \citet{zhu2024performance}, we derive a holds fulfillment method with an approximation ratio of $\frac12$ when optimizing usage (efficiency) alone. 

Simply applying tools developed in this line of work is insufficient for our goal, as we are also optimizing for browser experience, which is akin to the quality of the inventory at various locations \textit{itself}, as opposed to that utilized by customers. There has been attempts at this challenge, e.g., \citet{van2002variety} study the variety of assortment offered to customers, which alludes to the variety of the inventory available. However, to the best of our knowledge, directly optimizing for the inventory quality over time by incorporating it as either an objective or a constraint still remains an open problem.

Our approach augments theoretical modeling with empirical methods---we incorporate simulation-optimization methods, on a calibrated simulator, to measure objectives and optimize policies that are otherwise hard to model theoretically. However, the theoretical development is crucial: finding a $\frac12$-approximation fulfillment policy substantially reduces the parameter space over which the simulation-optimization must search. For simulation-optimization in our empirical application, we use recent advances in Bayesian optimization \citep{balandat2020botorch, daulton2021parallel}. 

Finally, our work connects to a large literature in \textit{urban bike sharing} operations, which asks: how should stations be distributed spatially \cite{he2021customer} and how should the operator rebalance bikes between stations (e.g., overnight \citep{jian2016simulation} or during the day \citep{freund2020data}), or incentivize users to do so \citep{Singla_Santoni_Bartók_Mukerji_Meenen_Krause_2015,chung2018bike}? While the applications share commonalities in managing spatially distributed resources that flow throughout the network over time, there are substantial differences: in our library application, resources (books) are highly heterogeneous, the operational lever includes choosing the \textit{source} location of the provided resource (as opposed to rebalancing), and the objective includes improving an in-person browser's experience. (We note that inventory \textit{rebalancing} at scale is not currently a design decision at the New York Public Library; when a user returns a book, it is sent back to the source library).

\section{Model and methods}


In this section, we introduce our model and main technical tools. Our goal is to jointly optimize (1) usage and (2) browser experience, using a combination of (a) hold requests fulfillment and (b) inventory reserve for browsers. 

In \Cref{sec:model_usage}, we introduce a model of library operations, where patrons arrive over time, and copies of the same book at different branches may be used to fulfill the same hold request. We first focus on optimizing usage with hold request fulfillment method, in which the source branch to fulfill each hold request is chosen. We introduce several fulfillment methods for maximizing usage in \Cref{sec:policy}, each being less optimal but more tractable. In \Cref{sec:browserreserve}, we then discuss the hardness of improving browser experience with inventory reserves through analytical tools and motivate the use of simulation-based methods.






\subsection{A model of library usage}
\label{sec:model_usage}

Consider a library system consisting of a set $B$ of branches, indexed by $i$. For brevity, in the model presentation we consider one title of books, multiple copies of which are available in multiple branches. (The optimization problem we will introduce is separable by book titles, and in the empirical application, we will optimize policies applied to thousands of titles simultaneously). We first present a model to characterize usage (checkouts) of books, by either holds or in-person patrons. In the next section, we develop a fulfillment policy to maximize usage under this model.

 Of the copies that are available at each branch $i \in B$, they are further separated into two types: (1) copies that are reserved for browser usage at that branch, and (2) copies that can be used by both browsers at that branch, \textit{and} can be used to fulfill holds at any branch. We use $a \in A$ to denote the type for each copy, where $A = \{(i,u): i\in B, u\in \{\text{reserve, non-reserve}\}\}$. Let $\textbf{c} = \{c_a, a\in A\}$ denote the vector of starting inventory of each type of copies.

Library patrons arrive over some time, and patron belong to different types, that distinguish the location that they enter the library system (their home branch), and their modes of checkout (browser or holds user). Concretely, we let patrons be indexed by types $j\in P$, where $P = \{(i,u): i\in B, u\in \{\text{browse, hold}\}\}$.\footnote{Naively, a person who both browses in person and uses the holds system can be represented as two patrons, and analogously for those who frequent multiple branches.}


We let $A_j \subset A$ be the subset of types of copies that can be used to satisfy the request of patron type $j$. For example, if $j = (i, \text{browse})$, then $A_j = \{(i, \text{reserve}), (i,\text{non-reserve})\}$: both the copies reserved and not reserved for browsers at branch $i$ could be used to satisfy this browser. On the other hand, if $j = (i, \text{hold})$, then $A_j = \{(i,\text{non-reserve}) \text{ for } i \in B\}$: we can only use copies that are not reserved for browsers, but these copies can be drawn from any branch in the system.

This stylized model of library usage captures the essence of the holds system: a network of branches where one type of patron demand can be satisfied by multiple types of copies in the inventory. In the next section, we focus on optimizing usage using a holds fulfillment method, given a set of initial inventory levels.




\subsection{Maximize usage via fulfillment}
\label{sec:policy}

In this section, we focus on maximizing usage through \textit{holds fulfillment}: when a request comes in and the book is available at multiple branches, from what branch will the request be served? We consider (1) a dynamic program-based fulfillment that optimizes a reward function to maximize overall checkouts, combining both browser and holds checkouts, in \Cref{sec:dppolicy}; (2) a near-optimal fulfillment method that achieves a $\frac12$ (worst case) approximation of the above dynamic program, in \Cref{sec:nearoptimal}; (3) a practically implementable, \textit{static} tiered fulfillment method that further approximates the near-optimal fulfillment method by sorting branches into one of three tiers, in \Cref{sec:tiered_fulfillment_sec2}.


Suppose each checkout from patrons of type $j$ yields a reward of $r_j\ge 0$.\footnote{Practitioners can use $r_j$ to reflect how much they value usage by different types of patrons. For example, $r_j=1$ for all $j$ would yield an optimization problem that maximizes an unweighted number of total checkouts. In our empirical application, we use non-homogeneous $r_j$ best reflects our metric of efficiency. This is discussed in detail in \Cref{sec:nearoptimalusagemaximizingsimulation}.}. We consider patron arrivals over a fixed time period $t\in \{1,2,\dots,T\}$, and during each period, with probability $\lambda_j\in [0,1]$, the arriving patron is of type $j\in P$.


\subsubsection{Usage-optimal dynamic program fulfillment}  
\label{sec:dppolicy}

Let $V_t(\textbf{x})$ denote the expected reward starting from period $t$ with some inventory level $\textbf{x} \le \textbf{c}$, and let $\textbf{e}_a$ denote a vector with the same length as $\textbf{c}$, with value $1$ at position $a$ and $0$ otherwise. The following dynamic program (DP) reflects the problem of maximizing our expected reward, considering just system usage:
\begin{align*}
   V_t(\textbf{x}) & = \sum_{j\in P} \underbrace{\left[\lambda_{j} \max_{a\in A_j: x_a > 0}\Big(r_j + V_{t+1}(\textbf{x} - \textbf{e}_a), V_{t+1}(\textbf{x})\Big)\right]}_{\text{Type $j$ patron arrives, use one copy and collect reward, or reject}} + \underbrace{\left(1-\sum_{j\in P} \lambda_j \mathbf{1}\left(\sum_{a\in A_j}x_a >0\right)\right) V_{t+1}(\textbf{x})}_{\text{No patron arrives or no available inventory}}\\
   & = \sum_{j\in P} \left[\mathbf{1}\left(\sum_{a\in A_j} x_a > 0\right)\left(r_j + \max_{a\in A_j: x_a >0} V_{t+1}(\textbf{x} - \textbf{e}_a) - V_{t+1}(\textbf{x})\right)^+ + V_{t+1}(\textbf{x})\right],\quad  \forall t\in [T], \textbf{x}\ge 0,
\end{align*}
with boundary conditions $V_{T+1}(\textbf{x}) = 0, \forall \textbf{x}$ and $V_{t}(\textbf{0}) = 0, \forall t$. The first part of the DP maximizes the expected reward from a potential patron's arrival---it makes the decision of whether to serve the patron with some type $a$ copy (and deduct this unit from the inventory)---while the second part carries the reward from next period in the case that no patron arrives or no inventory is available to serve the arriving patron. This DP captures the library's hold request fulfillment decisions in an online fashion, given the anticipation of patron arrivals (which can be estimated through historical data) and current inventory levels $\textbf{x}$. Solving this dynamic program would yield the optimal fulfillment method that maximizes usage. However, this involves calculating $V_{t}(\textbf{x})$ for an exponential number of possible states of the inventory $\textbf{x}$ at every time step $t$ (and for every book title in the library), thus rendering it intractable.

\subsubsection{Near-optimal fulfillment method given a value function approximation} 
\label{sec:nearoptimal}

We derive a computationally tractable fulfillment method that achieves a near-optimal (1/2 approximation guarantee relative to the DP) usage objective. This fulfillment method is based on a value function approximation to the linear programming relaxation of the DP. We apply the technical tools developed by \citet{zhu2024performance} to derive a $1/(1+L)$ approximate fulfillment method, where $L$ is the maximum number of resources used by a customer; in our problem, each patron only checks out $L=1$ copy of a book, giving a $\frac12$ approximation ratio.

\paragraph{Value function approximation} To derive such a near-usage-optimal fulfillment method, we first consider a linear program (LP) which provides an upper bound on the expected reward:
\setlength{\arraycolsep}{2pt}
\begin{subequations}
\label{eq:fluid-lp}
\begin{align}
    \text{{LP}} = \max_{\textbf{z}} \quad &\sum_{j\in P} r_j\sum_{a \in A_j} z_{ja}\\
    \text{s.t.} \quad  & \sum_{a\in A_j} z_{ja} \le T \lambda_j, & \hspace{-10em}\forall j \in P, \label{constr:lp1}\\
    & \sum_{j: a\in A_j} z_{ja} \le c_a, & \forall a \in A,\label{constr:inventory}\\
    & z_{ja} \ge 0, &\forall j\in P \text{ and } a\in A_j.
\end{align}
\end{subequations}
Intuitively, each decision variable $z_{ja}$ represents the expected consumption of type $a$ copies by type $j$ patrons. Constraint \ref{constr:lp1} ensures the expected consumption by type $j$ patrons is no larger than their expected arrivals $T\lambda_j$, and constraint \ref{constr:inventory} ensures the expected consumption of type $a$ copies does not exceed the starting inventory $c_a$.
\begin{proposition}
\label{prop:LPUB}
    $LP \ge V_1(\textbf{c})$, where $\textbf{c}$ is the vector of starting inventory levels.
\end{proposition}
\Cref{prop:LPUB} can be proven by showing that every sample path of the dynamic program can be converted into a feasible solution to the LP; this is a standard  approach in the literature, see e.g., \citet{gallego1994optimal}. The advantage of this LP is that it induces an approximate value function for the dynamic program. Consider the following approximation of value function $V_t(\textbf{x})$:

\begin{equation}
    H_t(\textbf{x}) = \sum_{a\in A} \hat{\gamma}_{at} x_a,
    \label{eq:value_func_approx}
\end{equation}
where coefficients $\hat{\gamma}_{at}, \forall a\in A$ are computed through \Cref{alg:one}. Intuitively, $\hat{\gamma}_{at}$ is the `expected reward per copy' we can get from having $x_a$ type $a$ copies, starting from time period $t$. The idea of \Cref{alg:one} is akin to inventory balancing: at each time period $t$ and for each patron type $j$, \Cref{alg:one} identifies the best type of copy to serve this patron type by finding the copy type with minimum $\hat{\gamma}_{a,t+1}$; it then increases this copy type's coefficient in period $t$ (which means that it is less likely to be the best type in period $t-1$). Appendix \Cref{lem:1} establishes that $H_1(\textbf{c})$ is at least $\frac{1}{2}$ of the LP, and thus, $H_1(\textbf{c})\ge \frac{1}{2}V_1(\textbf{c})$ -- $H_1(\textbf{c})$ provides a good approximation to the DP.

\begin{algorithm}[tb]
	\SetAlgoNoLine
	\KwIn{$\lambda_j, \forall j\in P$ and $c_a, \forall a \in A$.}
	\KwOut{Coefficients $\hat{\gamma}_{at}, \forall {a \in A, t \in [T]}$.}
	Initialize: $\hat{\gamma}_{a, T+1} = 0, \forall a\in A$\;
        \For{$t\in \{T, T-1,\dots,1\}$}{
            \For{$j\in P$}{
            Identify best type: $a^*_{jt} = \argmax_{a\in A_j} (r_j-\hat{\gamma}_{a, t+1})^+$\;
            }
            Update: $\hat{\gamma}_{at} = \hat{\gamma}_{a, t+1} + \sum_{j\in P}\frac{\lambda_{j}}{c_i} \mathbf{1}(a = a^*_{jt})(r_j-\hat{\gamma}_{a, t+1})^+$, for all $a\in A_j$\;
            
        }
	\caption{Calculate value function approximation coefficients}

    \label{alg:one}
\end{algorithm}






To present the near-optimal fulfillment method, we define indicator function $U_{ja}^t(\textbf{x}), \forall j\in P, a\in A_j, \textbf{x}\ge 0$ as follows:
\begin{equation}
    U_{ja}^t(\textbf{x}) = 
    \begin{cases}
        1, & \text{ if } r_j \ge H_{t+1}(\textbf{x}) - \max_{a\in A_j, x_a>0} H_{t+1}(\textbf{x} - \textbf{e}_a),\\ & \quad \text{and } a = \arg\max_{a\in A_j, x_a>0} H_{t+1}(\textbf{x} - \textbf{e}_a)\\
        0, & \text{otherwise}.
    \end{cases}
    \label{eq:indicator}
\end{equation}


Our proposed fulfillment method takes the following form:
\begin{quote}
    \textit{At state $\textbf{x}$ and time $t$, when a patron with type $j$ arrives, we provide them a copy of type $a$ if and only if $ U_{ja}^t(\textbf{x}) = 1$.}
\end{quote}
What does this method do? For each arriving patron, we will identify a type of copy that (1) can serve this patron, (2) has positive inventory, and (3) has the lowest benefit being in the inventory, and will fulfill this patron's needs if and only if the expected reward from this copy sitting in the inventory going forward is no greater than the immediate reward of serving this book to this patron.\footnote{In theory, such fulfillment would reject some patrons' demand for books. However, with well designed reward structures, this would not happen in practice. For example, if $r_j=1, \forall j$, the update rule in \Cref{alg:one} ensures that $\hat{\gamma}_{at}\le 1$, that is, the expected number of checkouts from a copy in the inventory is never larger than 1, and thus some copy will always be used to satisfy demand, if there exists an available copy in the system of the appropriate type.} We can then express the expected reward from this fulfillment method as
\begin{equation}
    R_t(\textbf{x}) = \sum_{j\in P} \lambda_j \sum_{a\in A_j} U_{ja}^t(\textbf{x})(r_j + R_{t+1}(\textbf{x} - \textbf{e}_a)) + \left(1-\sum_{j\in P} \lambda_j \sum_{a\in A_j} U_{ja}^t(\textbf{x})\right) R_{t+1}(\textbf{x}), \quad  \forall t \in [T], \textbf{x}\ge 0,
\end{equation}
with boundary conditions $R_{T+1}(\textbf{x}) = 0, \forall \textbf{x}\ge 0$ and $R_{t}(\textbf{0}) = 0, \forall t\in [T]$. The first part of $R_t$ gives us the expected reward in the case that a book is used in period $t$, whereas the second part carries the reward from the next period in the case that no books are used. 

We show that this fulfillment achieves at least $1/2$ of the usage objective of the DP solution.
\begin{restatable}{proposition}{mainthm}
\label{thm:main}
    $R_1(\textbf{c}) \ge \frac{1}{2}V_1(\textbf{c}).$
\end{restatable}

The proof first relies on \Cref{lem:1}, which establishes that $H_1(\textbf{x})$ approximates the DP within a factor of $\frac12$, and then proceeds through induction, leveraging both the structure of the indicator function in \Cref{eq:indicator} and \Cref{alg:one}. Finally, we note that, although the $\frac12$ approximation guarantee is the best provable, prior works have shown that in practice, similar algorithms achieve performance much closer to the LP upper bound, and thus the DP (see e.g., Section 5.3 in \citet{ma2020approximation}). This fulfillment gives us a principled, tractable, and provably efficient way to fulfill hold requests. As we will show, it can be further translated into an \textit{offline} fulfillment method that is compatible with NYPL's existing infrastructure and achieves almost identical performance in simulation.

\subsubsection{Practically implementable tiered fulfillment}
\label{sec:tiered_fulfillment_sec2}
The policies introduced above are online, dynamic policies: the source destination for a request can vary at the individual request level, depending on the number of copies of that book available at each branch at the time of the request. In contrast, the New York Public Library currently uses a static, offline policy that is shared across books and branches, in what is known as a `tiered paging list' policy: each branch is assigned 1 out of 3 tiers; whenever a hold request appears at a branch (by a patron who would like to pick up a book -- that is currently available elsewhere in the system -- at a given branch), the system will call branches which have the title available in order -- it will first attempt to fulfill the request with a copy from one of the tier-1 branches, and move down to tier-2 and tier-3 branches respectively, until the request is fulfilled.\footnote{Practitioners point out other nuances that may be involved, such as irregularities in the inventory system, or branch staff unavailable to pull holds physically. These are aspects hard to capture by any analysis.} Thus, we further approximate the above near-optimal policy, by sorting libraries into tiers, informed by the above policy.

Note that parameter $\hat{\gamma}_{at}$ in the value function approximation  $H_t(\textbf{x})$ (in \Cref{eq:value_func_approx}) corresponds to the `expected reward per copy in inventory' for type $a$ books in period $t$ of the policy's planning horizon, and near-optimal fulfillment prioritizes utilizing types that have low values of $\hat{\gamma}_{at}$. To calculate an approximation, for each branch $i$, and for a set of book titles $l\in L$, we calculate the following quantity:
\begin{equation}
\label{eq:implementable_approx}
\frac{\sum_{l\in L} \textbf{1}\left(c_{(i,\text{non-reserve})}^{{l}}>0\right) \hat{\gamma}_{(i, \text{hold}), 1}^{{l}}}{\sum_{{l\in L}} \textbf{1}\left(c_{(i,\text{non-reserve})}^{{l}}>0\right)},
\end{equation}
where $\hat{\gamma}_{(i, \text{hold}), 1}^{l}$ is the expected reward generated by a copy of a given book title at branch $i$, at the start of the planning horizon (t=1); and $c_{(i,\text{non-reserve})}^{{l}}$ is the initial number of copies of book title $l$ at branch $i$ that are not reserved for only browser usage. This measures the `average unit reward' from a copy of any book at branch $i$ (not reserved for browser usage) at the start of the horizon. We then segment all branches into three equal groups, with low (middle, high) average unit reward branches assigned to tier 1 (2, 3), respectively.

After this assignment of branches to static tiers, in our simulation described below, whenever a hold request comes in, we mimic what is done currently: attempt to fulfill this request with a copy from a branch in the first tier, and move down to the second (subsequently, third) tier list only if an available copy could not be found within branches of the previous tier. As we will show in the empirical results,  perhaps surprisingly, each method in this class of implementable tiered fulfillment methods performs close to the near-optimal fulfillment that it is designed to approximate.\footnote{We note that, given different initial non-reserve inventory levels $c_{(i,\text{non-reserve})}^{{l}}$, influenced by the browser reserve fractions discussed in \Cref{sec:browserreserve}, the tier assigned to the same branch could vary across policies; thus, tiered fulfillment corresponds to a \textit{class} of fulfillment methods.}

\subsection{Balance browser experience with browser reserve inventory}
\label{sec:browserreserve}
Above, we show that there is a principled way to maximize overall usage by optimizing holds fulfillment, given a set of initial inventory levels. In this section, we briefly discuss varying inventory reserved for browsers to balance the need for a better browser experience, and why this necessitates simulation optimization methods. A product in a warehouse, or a flight seat serves no utility until they are sold. Diverse and high-quality books on a library shelf, on the other hand, provide value for in-person library patrons, who may read or learn about books while at the branch. A browser experience objective should thus measure browsers' exposure to the quality of books on the library shelf. However, such an objective would be challenging, if not intractable, to incorporate into the above theoretical framework, as it requires taking as an objective the integral of unused inventory through each time period.

$$
\{\beta^*, \hat{\mathcal{F}}^*\}
$$

Our main lever to balance browser experience will be through varying inventory levels \textit{reserved} for browsers. In our model, each branch has its predetermined number of copies of each title in its collection, and the system can choose some of these copies to be set aside so that they are only available for browsers (to checkout or read in-person) at this branch, and not available through the holds system for patrons at other branches. This is equivalent to varying $c_{(i, \text{reserve})}$ and $c_{(i, \text{non-reserve})}$ for each branch $i$ to optimize an overall browser experience objective, subject to $c_{(i, \text{reserve})} + c_{(i, \text{non-reserve})}\le c_{i}$ for some predetermined $c_i$.


Instead of optimizing for specific starting inventory levels for each branch (which would become intractable when optimized for each title), we introduce a parameter $\beta_i\in [0,1]$ for each branch $i\in B$, referred to as ``browser reserve fraction''. This parameter captures, on average, how much of the total inventory of each branch is reserved for browser usage. At the two extremes, $\beta_i=0$ means that no books are reserved at this branch (i.e., they can all be used to fulfill hold requests at other branches; this is (largely) the status quo for NYPL), and $\beta_i=1$ means that all books are reserved for browser reserve, and no patrons from other branches can request books from this branch.

Optimizing such reserves theoretically is challenging, even after we reduce it to optimizing one $\beta_i$ for each branch $i$. To the best of our knowledge, incorporating joint inventory optimization into the problem of revenue management with flexible products has not been studied in the literature. To overcome the hardness in measuring browser experience and optimizing browser reserve inventory, we next introduce our simulation optimization framework.

\section{Simulation optimization framework}
In the previous section, we outline technical tools and levers for optimizing usage and browser experience in a library system. In this section, we introduce our simulation optimization framework, that we will use to optimize the policy levers and evaluate the resulting designs. We start by an overview of the simulator itself, and how we calibrate it to historical data in \Cref{sec:overviewofsimulator}. In \Cref{sec:sim_measure_obj}, we detail how usage and browser experience objectives are measured in the simulation, which overcomes the difficulty of measuring browser experience through stylized modeling. In \Cref{sec:sim_fulfillment_policy}, we discuss how the near-optimal fulfillment method (denoted by $\mathcal{F}^*$) and the implementable tiered fulfillment (denoted by $\hat{\mathcal{F}}^*$) are implemented, given a set of initial inventory levels. Finally, in \Cref{sec:sim_opt}, we outline how the levers are optimized through simulation optimization.

\begin{algorithm}[tb]
\caption{Simulation process}
\label{alg:sim}
	\SetAlgoNoLine
	\KwData{Patron arrival probabilities $\lambda_j, \text{ for all patron types }j$; starting inventories $c_i, \text{ for all branches }i$; number of days to simulate.}
        \KwIn{Policy parameters: browser reserve fractions $\beta_i,$ for all branches $i$; fulfillment method $\mathcal{F}$.}
	\KwOut{Usage objective $f(\vec{\beta}, \mathcal{F})$ and browser experience objective $g(\vec{\beta}, \mathcal{F})$.}
        Initialize inventory: $c_{i, \text{reserve}} \sim \text{Binomial}(c_i, \beta_i), c_{i, \text{non-reserve}} = c_i - c_{i, \text{reserve}}$\;
        \For{each day in simulation}{
            Generate patron arrivals according to $\lambda_j$, and generate book returns\;
            Fulfill patron requests according to $\mathcal{P}$, given current inventory levels\;
        }
\end{algorithm}

\subsection{Overview of the simulator}
\label{sec:overviewofsimulator}

Our simulator first generates patron arrivals, and then simulates checkouts based on a policy $(\vec{\beta}, \mathcal{F})$, which consists of a set of browser reserve fractions $\vec{\beta}$ and a hold fulfillment method $\mathcal{F}$, for a given period of time. At the end, we collect the sample path of checkouts and inventory levels, and calculate usage objective $f(\vec{\beta}, \mathcal{F})$ and browser experience objective $g(\vec{\beta}, \mathcal{F})$. \Cref{alg:sim} provides an overview of this process.

A challenge in enabling simulation realism is generating patron arrivals that mimic reality. One possible solution would be to use historical checkouts and returns as direct input into the simulator. However, this solution does not allow for the flexibility of varying the level of inventory reserved for browsers. Historically, almost all books are open to being used to fulfill hold requests,\footnote{With the exception of a small number of popular books reserved at a few branches, known as the `Lucky Day collections.'} and by varying reserve levels, historical checkouts may become infeasible due to inventory being no longer available; i.e., we would quickly be ``off-policy'' in the simulation. 

Instead, we estimate arrival probabilities from historical data, and then build a simulator based on these arrival probabilities, calibrating the simulator to historical checkouts when using historical inventory reserve levels (i.e., no reserve) and a fulfillment method that best resembles historical practices. The estimation builds off of a Bayesian latent variable model that decomposes arrival rates into three components, each ranging from 0 to 1: (1) $p_i$: the magnitude of demand size at branch $i$, (2) $h_i$: the fraction of checkouts being from the holds system at branch $i$, and (3) a measure of book desirability $d$. We use the parameters estimated by \citet{liu2024identifying}. As an example, for a book with desirability $d$, the arrival probabilities of patrons at branch $i$ that are either browsers or hold users would be
\begin{equation}
    \lambda^d_{(i,\text{hold})} \approx d p_i h_i, \quad \lambda^d_{(i,\text{browser})} \approx d p_i (1-h_i). \footnote{We make further adjustments to these estimates, so the resulting outcome of the simulation under a policy that mimics historical practice could better match historical checkout data.}\label{eq:browserarrival}
\end{equation}
In other words, on each day of the simulation, with probability $\lambda^d_{(i,\text{hold})}$, a hold user comes into branch $i$ and requests to check out a copy, and with probability $\lambda^d_{(i,\text{browser})}$, a browser comes into branch $i$ and attempts to check out a copy. For NYPL, after a book is checked out, there is a 21-day period to read the book. We thus simulate a return after a book has been checked out for 21 days in the simulator. 
\citet{liu2024identifying} show that this method accurately recovers historical checkouts of each book title at each branch. In our empirical application, we further show that our simulation of the historical policy under these parameters recovers the objective function values based on the true historical number of checkouts of each title at each branch.


\subsection{Usage and browser experience objectives}
\label{sec:sim_measure_obj}

One major benefit of a simulation framework is that it allows for measuring any metric based on a sample path of the simulation, including metrics that may be intractable for a stylized model. For a policy, we measure usage based on the (weighted) number of checkouts at each branch and measure browser experience using the (weighted) cumulative sum of the quality of books available at each branch during the entire simulation period. 

\subsubsection{Usage objective} Concretely, for each branch $i\in B$, let $\CO_i(\vec{\beta}, \mathcal{F})$ denote the total checkouts (both browser and hold) at branch $i$, with reserve fractions $\vec{\beta}$ and hold fulfillment method $\mathcal{F}$. 

The policy that (almost)\footnote{Not exactly, given the approximation gap from the dynamic program and differences between the theoretical setup (including fixed time horizon, non-reusability) and the simulation.} maximizes total checkouts over all branches, i.e., $\sum_{i\in B}\CO_i(\vec{\beta}, \mathcal{F})$, would be if no branch has browser reserve inventories ($\vec{\beta} = \textbf{0}$) and we use the near-optimal fulfillment $\mathcal{F}^*$ -- this way, the entire collection of the library is maximally available for checkouts, and we are handling hold fulfillment efficiently. For each input $(\vec{\beta}, \mathcal{F})$, we calculate ratios 
\begin{equation}
    \frac{\CO_i(\vec{\beta}, \mathcal{F})}{\CO_i(\textbf{0}, \mathcal{F}^*)}
    \label{eq:usage_ratios}
\end{equation} 
for each branch $i\in B$. This ratio captures the gap between usage levels under policy $(\vec{\beta}, \mathcal{F})$, compared to the `ideal' usage level under policy $(\textbf{0}, \mathcal{F}^*)$.\footnote{Note that $CO_i(\textbf{0}, \mathcal{F}^*)$ is not, even approximately, an upper bound on the usage level for \textit{each} branch $i$---we would only expect the upper bound to approximately hold for the system as a whole. For example, if we reserve all inventory at branch $i$ for browsers, reserve no inventory at other branches, and use near-optimal fulfillment: then, more browser checkouts and a similar number of hold checkouts at branch $i$ would occur than under the policy $(\textbf{0}, \mathcal{F}^*)$. Empirically, we observe some branches to have ratios greater than $1$ for some policies, but the overall usage objective $f$ does not exceed 1, as expected.}

Finally, we calculate the usage objective for $\vec{\beta}\ge 0, \mathcal{F\in \{{\mathcal{F}}^*,\hat{\mathcal{F}}^*\}}$ as
\begin{equation}
\label{eq:usage_objective}
f(\vec{\beta}, \mathcal{F}) = \frac{\sum_{i\in B} p_i \frac{\CO_i(\vec{\beta}, \mathcal{F})}{\CO_i(\textbf{0}, \mathcal{F}^*)}}{\sum_{i\in B} p_i},
\end{equation} where recall that $p_i$ is a branch demand size parameter. Thus, this measure captures is latent demand size-weighted average usage ratio of each branch.  We use this approach, instead of the sum over checkouts across branches directly, so that no branches are fully ignored in the objective, while still weighing larger branches more heavily.

\subsubsection{Browser experience objective} 


We now define a metric for the quality of a branch's on-shelf collection through time, as a measure of browser experience experience. 

Concretely, for a set of book titles $l\in L$, let $D_i^l(\vec{\beta}, \mathcal{F})$ denote the number of days book title $l$ is available in branch $i$ over some period of operation, under policy $(\vec{\beta}, \mathcal{F})$. We measure the total \textit{collection quality} of a branch as
$$
\CQ_i(\vec{\beta}, \mathcal{F}) = \sum_{l\in L} d_l D_i^l(\vec{\beta}, \mathcal{F}), \forall i\in B,
$$
where $d_l$ is the latent desirability score of title $l$. In other words, $\CQ_i(\vec{\beta}, \mathcal{F})$ measures the cumulative quality of a branch's collection over a long time period. 

The set of browser reserves that maximizes total on-shelf collection quality across branches is $\vec{\beta} = 1$.\footnote{When all the books are reserved for browser usage, which hold fulfillment method we choose is irrelevant, as no hold requests can be fulfilled.} That is, when all branches reserve their collection to only browser usage, we effectively disable the holds system and maximize collection quality. Thus, for each $(\vec{\beta}, \mathcal{F})$, we calculate ratios 
$$
\frac{\CQ_i(\vec{\beta}, \mathcal{F})}{\CQ_i(\textbf{1}, \mathcal{F}^*)}
$$ for each branch $i\in B$, which captures the gap between collection quality under policy $(\vec{\beta}, \mathcal{F})$, compared to the `ideal' collection quality without a holds system.

Finally, our browser experience objective is defined as
$$
g(\vec{\beta}, \mathcal{F}) = \frac{\sum_{i\in B} (1-h_i) \frac{\CQ_i(\vec{\beta}, \mathcal{F})}{\CQ_i(\textbf{1}, \mathcal{F}^*)}}{\sum_{i\in B} (1-h_i)},
$$
where $1-h_i$ is the browser usage fraction of branch $i$. In words, we are measuring the browser usage-weighted average of the collection quality ratio of each branch.\footnote{We also conduct additional analyses on when $g$ takes the form of Nash welfare function (i.e., geometric mean) of the ratios, and the resulting empirical trade-off is similar to results presented in the main text. See \Cref{sec:nash_app}.} When $\vec{\beta}=\textbf{1}$, we have $g(\textbf{1}, \mathcal{F}^*) = 1$, and with decreasing $\vec{\beta}$, $g$ tends to also decrease -- lower reserve fractions lead to higher usage, and thus a smaller collection on shelves.

\subsection{Two classes of policies}
\label{sec:sim_fulfillment_policy}

In this section, we introduce the two policy classes considered in our simulation framework. The first class of policies $(\vec{\beta}, \mathcal{F}^*)$, referred to as \textit{near-optimal usage-maximizing policies}, consists of a set of branch-specific browser reserve fractions $\vec{\beta}\in [0,1]^{|B|}$, and uses near-optimal hold fulfillment method $\mathcal{F}^*$. The second class of policies $(\vec{\beta}, \hat{\mathcal{F}}^*)$, referred to as \textit{implementable tiered policies}, similarly consists of a set of browser reserve fractions but uses the implementable tiered hold fulfillment method $\hat{\mathcal{F}}^*$.



For both classes of policies, we initialize inventories using reserve fractions ${\beta_i}$ and the total inventory levels $c_i$: at each branch $i$, $c_{i,\text{reserve}}$ copies are initialized to be reserved for browsers, and $c_{i, \text{non-reserve}}$ can serve both browsers and hold users, where 
$$
c_{i,\text{reserve}}\sim \text{Binomial}(c_i, \beta_i),
$$ 
and $c_{i, \text{non-reserve}} = c_i - c_{i,\text{reserve}}$.\footnote{We randomize this browser reserve assignment in our simulation, because for the majority of book titles, only very few copies are available at each branch, and a simple deterministic rounding of $c_i\beta_i$ would mean that the inventory levels, and in turn the sample paths are highly non-smooth as a function of $\beta_i$, making optimization challenging. We note that in practice, setting browser reserves could be done directly on NYPL's digital management system, which makes this random assignment or another rounding approximation -- such that rounding across titles in the branch -- feasible.} The two classes of policies differ in how they fulfill hold requests.

\subsubsection{Near-optimal usage-maximizing policy} 
\label{sec:nearoptimalusagemaximizingsimulation}

To simulate near-optimal usage-maximizing fulfillment, we first calculate coefficients $\hat{\gamma}$ according to \Cref{alg:one} and the inventory levels initialized, and then dynamically allocate books to satisfy hold requests. We discuss two aspects of the fulfillment method not addressed in the modeling section, but are crucial for implementation: the design of reward $r_j$ for a checkout by each patron type $j$, and the selection of a time horizon $T$.

\paragraph{Design of reward $r_j$} Reward $r_j$ for a book checked out by a patron of type $j$ is crucial in defining usage and calculating coefficients $\gamma$ in \Cref{alg:one}. Our usage objective (\Cref{eq:usage_objective}) naturally yields a reward structure. In \Cref{eq:usage_objective}, the contribution of one checkout at branch $i$ to the efficiency objective is
\begin{equation}
\frac{p_i}{\sum_{i\in B} p_i}\frac{1}{\CO_i(\textbf{0}, \mathcal{F}^*)}, \forall i \in B,
\end{equation}
which is the coefficient on $\CO_i(\vec{\beta}, \mathcal{F})$ in \Cref{eq:usage_objective}. We note that a checkout at branch $i$ could be attributed to two types of patrons: $j\in \{(i,\text{browser}), (i,\text{hold})\}$. We set $r_j$ in this form for both these two types of patrons, and repeat for every branch $i\in B$. This gives us the desired near-optimal fulfillment that seeks to maximize the exact usage objective we are measuring.

\paragraph{Selection of time horizon $T$} Books in a library, unlike seats on a flight, are \textit{reusable}. Patrons return a book to their branch after reading for a period of time,\footnote{All copies of books are assigned a home branch - a branch where that copy is supposed to be in when it is returned. Regardless of whether a copy is checked out via hold (which can happen at a branch different from the home branch) or browsing, eventually, it is returned to that home branch.} and that exact same copy of the book can then be used to serve other patrons. Tackling policy design challenges for such reusable resources is still an open question (e.g., \citet{zhu2024performance} provide an analysis extending the flexible fulfillment method to the case where each resource is reused at most twice, and the policy becomes much more complicated). We circumvent this reusability issue by selecting a time horizon $T$ within which the inventory is approximately non-reusable: before the book has been returned. To calculate coefficients $\hat{\gamma}$, we choose $T=21$, which is the period that patrons are allowed to keep the book once checked out, and recompute every 21 days with the inventory level at that point. Intuitively, this means that in the long-run average, the number of books returned during any 21-day period is at most the number of books in the inventory at the start of that period, so the total number of books checked out is at most twice the starting inventory. This mimics each book being reusable for up to twice, which as \citet{zhu2024performance} show, approximates the non-reusable case well in their setting.


\subsubsection{Implementable tiered policy} To simulate implementable tiered fulfillment, we take a given near-optimal fulfillment method, and determine the tiers of each branch by the methods outlined in \Cref{sec:tiered_fulfillment_sec2}. 
The essential step in this approximation is in calculating the `average unit reward' defined by \Cref{eq:implementable_approx}. To adapt to the simulation (where instead of calculating $\gamma$ once, they are re-calculated once 21 days), we take the inventory levels and coefficients $\gamma$ calculated immediately after the warm-up period of the simulation, evaluate \Cref{eq:implementable_approx}, and set static tiers for each branch. This design choice is motivated by two reasons. First, inventory levels after the warm-up period are relatively stable, thus taking into account all the recalculated $\gamma$'s does not affect the assigned tiers substantially. Second, this calculation mimics the non-clairvoyant and static nature of practical implementation in NYPL, where we would only have access to coefficients in the past, and do not wish to adjust the tiers frequently.

\subsection{Simulation optimization}
\label{sec:sim_opt}
Our simulator can be seen as a mapping from policy $(\vec{\beta}, \mathcal{F})$ to usage objective $f$ and browser experience objective $g$. We aim to find Pareto-optimal policies in the usage-browser experience space. For our empirical application, we leverage multi-objective Bayesian optimization techniques under the BoTorch \citep{balandat2020botorch} framework, and more specifically, we use the parallel noisy expected hypervolume improvement acquisition function \citep{daulton2021parallel} to generate new parameters (policies) with expected improvement. This allows for fast parallel evaluations and policy updates without gradient information.

Our simulation optimization approach takes two steps. First, we optimize browser reserves within the class of near-optimal usage-maximizing policy. Mathematically, we aim to find a Pareto optimal set of browser reserves:
\begin{equation}
    \{\vec{\beta}^*\} = \argmax_{\vec{\beta}}\ \alpha f(\vec{\beta}, \mathcal{F}^*) + (1-\alpha) g(\vec{\beta}, \mathcal{F}^*), \quad \forall \alpha \in [0, 1].
\end{equation}
Note that we fix the fulfillment method to be $\mathcal{F}^*$ here, thus restricting to the near-optimal usage-maximizing policy class: the theoretical analysis allows us to restrict the parameter space over which to search.

In the second step, for each vector of reserve fractions found on the Pareto frontier, we evaluate its performance under the implementable tiered policy class. Concretely, we evaluate $f(\vec{\beta}, \mathcal{\hat{F}}^*)$ and $g(\vec{\beta}, \mathcal{\hat{F}}^*)$ through simulation, for each $\vec{\beta} \in \{\vec{\beta}^*\}$. 

\section{Empirical application}

\label{sec:results}
\begin{figure}[tb]
    \centering
    \includegraphics[width=0.97\linewidth]{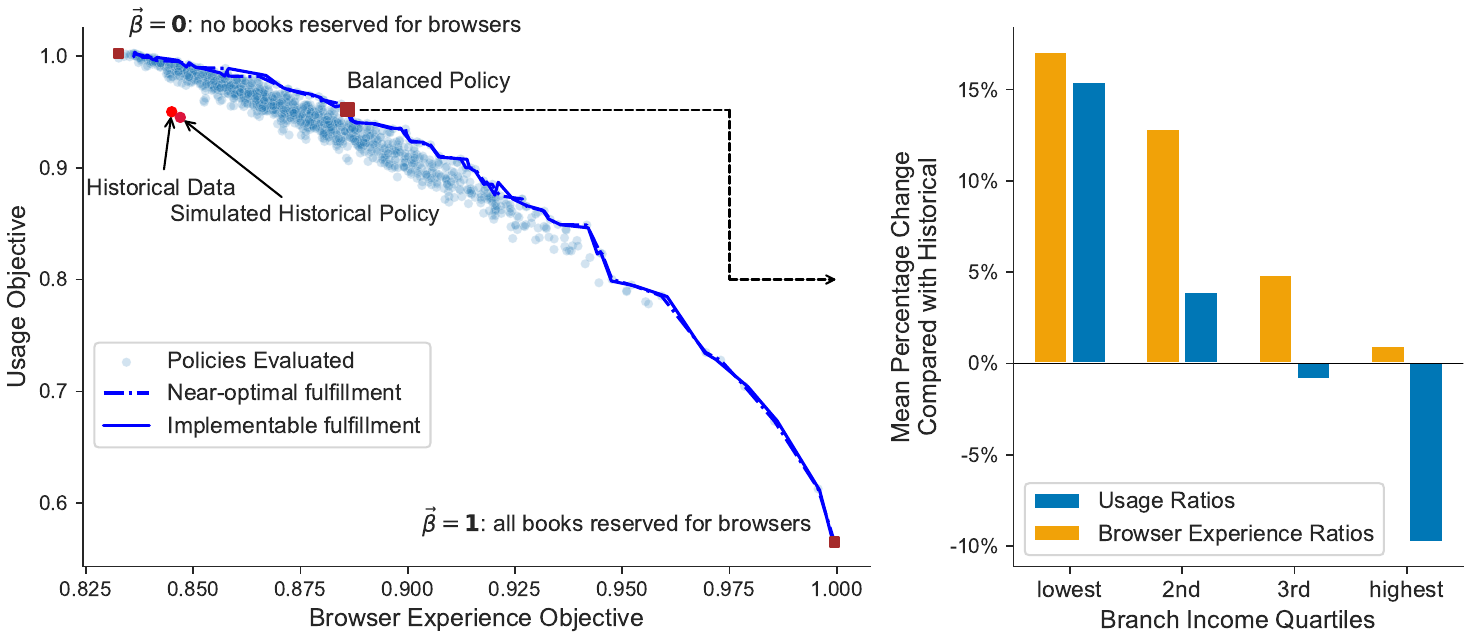}
    \caption{(Left) Pareto curve of the usage and browser experience objectives, based on 1,600 policies evaluated through the Bayesian optimization procedure. For each of the policies on the Pareto frontier, we additionally evaluate their implementable version. Historical data and a simulation of a policy approximating the historical policy achieve similar objective values. Crucially, we find many policies that Pareto dominate historical data. (Right) Examining one policy (denoted \textit{Balanced}) in detail, we find that it achieves higher browser experience generally across branches and more overall usage in branches in lower-income neighborhoods, at the cost of some usage at higher-income branches. These patterns reflect the balanced policy reserving some book copies for in-person browsers, who are disproportionately patrons of branches in lower-income neighborhoods.}
    \label{fig:pareto}
\end{figure}


We apply our approach to optimize the browser reserve and hold fulfillment for the New York Public Library (NYPL). We first present the Pareto frontiers of policies, and then focus on a policy that balances usage and browser experience well, and discuss its implications for different branches. 

\subsection{Data and computational methods}

The input to our simulator are arrival rates $\lambda_j, j\in P$ estimated using historical data from NYPL in calendar year 2022. The starting inventory levels $c_i, i\in B$ are set to be the total number of copies available at each branch $i$. For each simulation run, we start with a 100-day warm-up period, after which usage stabilizes, and then collect sample paths on checkouts and collection quality. We include 3,809 book titles across 84 branches in the simulation. These book titles are available at 20 or more branches, and though they only make up 1.27\% (out of 300,389) of all book titles within NYPL, in 2022, checkouts of these books account for 31.01\% (433,998 out of 1,399,351) of total checkouts. One salient difference with adding more titles is that there may be a sharper trade-off between holds usage and browser experience for titles with fewer copies. For branch and title level parameters (overall brach demand, hold usage fraction, and title desirability), we use the values calculated by \citet{liu2024identifying}, and adjust them to better calibrate to historical data. 

We run Bayesian optimization in batches of 16 policies, updating the acquisition functions between each batch. We initialize the optimization using a set of 16 policies in which the browser reserve fractions are set to a grid between 0 and 1, uniformly for all branches. This initialization allows the optimizer to explore the extremes of the parameter space, even though some of the extremes are sub-optimal. Each batch iteration of the Bayesian optimization procedure (which includes optimizing the acquisition function to get candidate policies, and then evaluating a batch of candidate policies through simulation) takes 15 to 60 minutes using 64 CPU cores and 512GB of RAM. We find that after 100 iterations, marginal improvements of new policy evaluations are small, indicating the policies found are close to the desired Pareto frontier. Data and methods are further discussed in \Cref{app:backgroun_application}. 



\subsection{Results}

We now analyze the results of our application. \Cref{fig:pareto} shows the Pareto frontier of usage and browser experience; it further compares the outcomes for a particular balanced policy to that of the historical policies, for branches in neighborhoods with varying average income levels.\footnote{This balanced policy we present is within the class of implementable tiered policy. The exact parameters of this policy are presented in Appendix \Cref{fig:policy_params}.} \Cref{fig:implications} breaks down these implications for individual branches, again for the balanced policy analyzed above, compared to the historical outcomes. Finally, \Cref{fig:focusonebranch} zooms in on one branch in particular, showing how its net outflows to different branches and shelve stocks differ between the balanced and historical policies. We give an overview of the insights from these analyses.

\paragraph{A substantial trade-off exists between usage and browser experience} In \Cref{fig:pareto}, at the two extremes, we have when no books are reserved for browsers ($\vec{\beta}=\textbf{0}$; roughly corresponding to the status quo, but with optimal fulfillment) and when all books are reserved for browsers ($\vec{\beta}=\textbf{1}$; effectively eliminating the holds system). From $\vec{\beta}=\textbf{0}$ to $\vec{\beta}=\textbf{1}$, we see more than 40\% drop in usage, and in the other direction, we see approximately 17\% drop in browser experience. On the one hand, this indicates that the holds system is indeed effective in improving usage, and is a useful addition to library systems that do not have it in place.

\begin{figure}[tb]
    \centering
    
    \includegraphics[width=.95\textwidth]{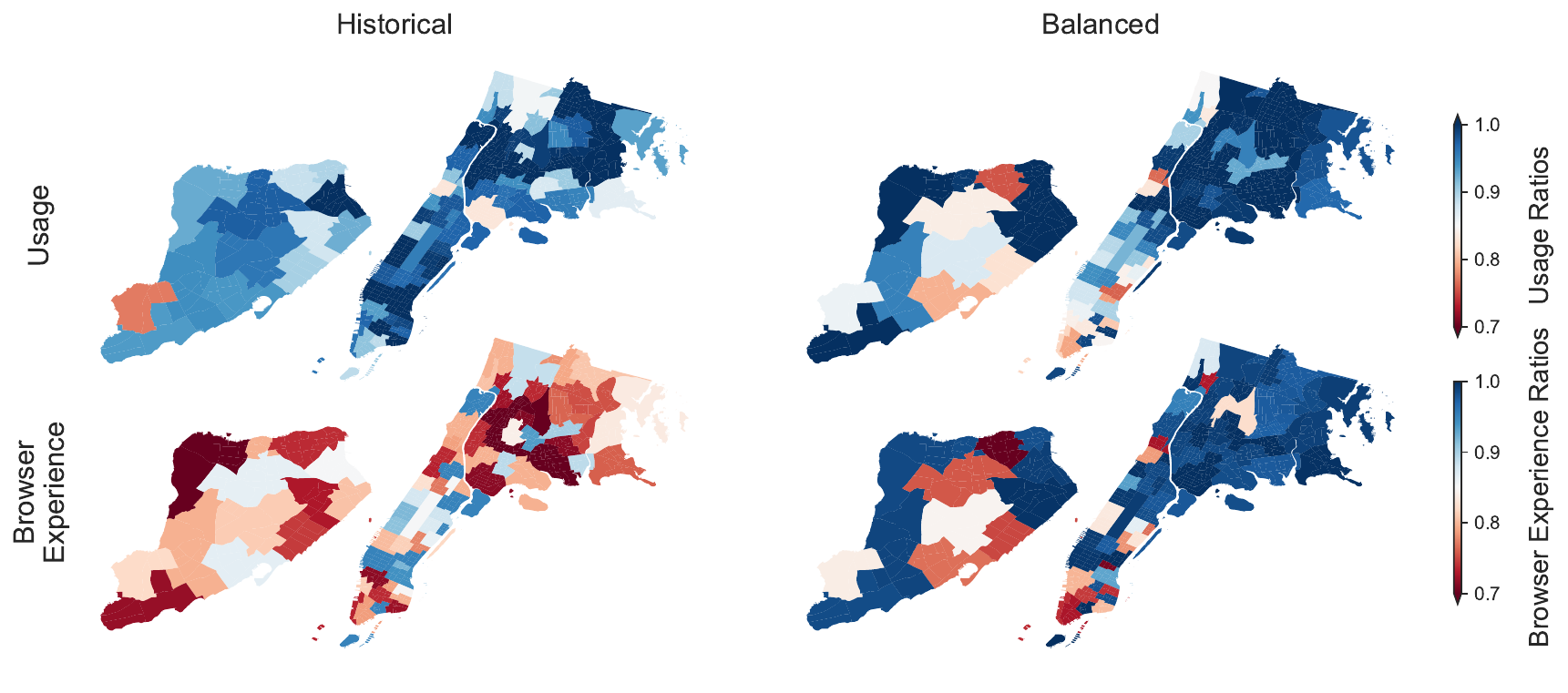}
    \caption{Map of the usage and browser experience objectives at each branch, for each of the historical data and under the Balanced policy (ratios between the usage and browser experience values for the respective policy and the `optimal' benchmark for that objective). We find that (1) the Balanced policy preserves usage for most branches compared to the Historical policy, except for some branches in lower Manhattan (lower tip of the center part of the map); (2) the balanced policy improves browser experience generally across branches, and especially branches in the Bronx (upper right), which are predominantly serving lower-income neighborhoods, who rely on browsing more (have lower hold usage). This is consistent with the Pareto curve in \Cref{fig:pareto}, which shows that the historical policy relatively prioritized overall usage, at the expense of the browser experience.}
    \label{fig:implications}
\end{figure}

\begin{figure}[tb]
    \centering
    
    \begin{subfigure}[b]{0.68\textwidth}
        \centering
        \includegraphics[width=\textwidth]{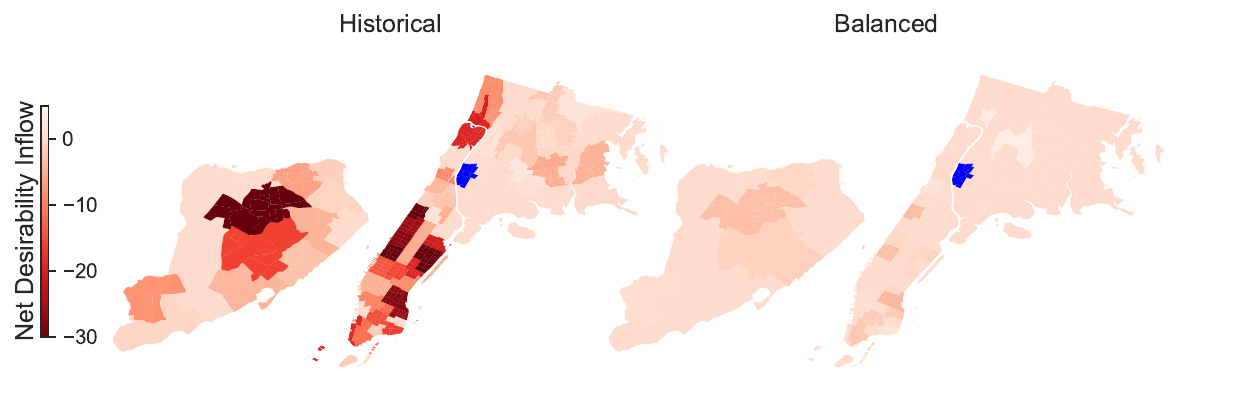}
        \caption{Comparison of net inflow patterns}
        \label{fig:netinflowchangeonebranch}
    \end{subfigure}
    \begin{subfigure}[b]{0.31\textwidth}
        \centering
        \includegraphics[width=\textwidth]{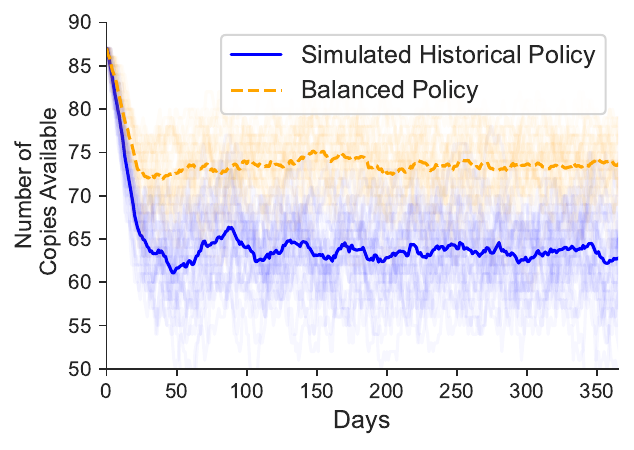}
        \caption{Comparison of inventory levels}
        \label{fig:inventorychangeonebranch}
    \end{subfigure}
    \hfill
    \caption{We study the impact of the balanced policy on one branch in particular: High Bridge Library in the Bronx, which serves an area highlighted in blue in subfigure (a). In the calendar year 2022, this moderately sized branch was among the lowest hold usage branches, and incurred the 6th highest net desirability outflow, as shown in \Cref{fig:1c}. As illustrated in (a), historically, books from this branch are used to fulfill hold requests from branches predominantly in midtown and downtown Manhattan, and the Upper West and Upper East Sides (dark red areas in the middle of the plot); in contrast, under the Balanced policy, it incurs much less overall outflow, with an evenly distributed net outflow pattern. We further study the total inventory levels of 60 of the most popular book titles (roughly 2\% of the total number of book titles simulated) at this branch. In subfigure (b), we compare the inventory levels over 30 simulation runs under the Balanced policy and approximate historical policy. The Balanced policy achieves higher inventory levels at equilibrium, providing a better browser experience at this branch.}
    \label{fig:focusonebranch}
\end{figure}

Historical policy represents a 5\% gap to optimal usage and a 15\% gap to optimal browser experience -- it is more aligned with optimizing usage rather than browser experience. This is no surprise, as the historical policy is essentially a combination of no browser reserve ($\vec{\beta}=\textbf{0}$) and sub-optimal hold fulfillment. This gap in browser experience disproportionately affects patrons in lower-income neighborhoods. 

\paragraph{A balanced policy achieves better browser experience without substantial impact on overall usage.} In \Cref{fig:pareto}, we find an array of policies that each Pareto dominate the historical policy. Taking a closer look at one of the policies that appear to achieve a good balance between usage and browser experience, we find that compared with historical data, it improves browser experience generally across branches and increases usage in branches in lower-income neighborhoods, at a relatively small cost of a decrease in usage in higher-income neighborhoods. This is a result of reserving higher levels of book copies for in-person browsers at lower-income branches, as well as assigning them lower tiers for hold fulfillment, as shown in Appendix \Cref{fig:policy_params}.

Mapping out these impacts spatially in \Cref{fig:implications}, one major improvement we see is in the borough of the Bronx, shown in the upper right corner, where a balanced policy achieves improved browser experience in all of the 35 branches in this borough. Historically lower-income and diverse, patrons from the Bronx rely on service from public libraries and particularly rely on in-person browsing. This balanced policy helps NYPL better serve patrons here. 

\paragraph{Balancing browser experience induces lower outflow and higher inventory levels at branches with high browser usage.} Focusing on one branch in particular, the High Bridge Library in the Bronx, we find more nuanced improvements that the balanced policy brings. As presented in \Cref{fig:netinflowchangeonebranch}, historically, much of the collection from High Bridge is used to satisfy hold requests from other branches, which not only sacrifices browser experience but also places a logistical burden on library staff in this modest branch.\footnote{Library staff need to physically retrieve a copy of the book requested, and place it on a holds shelf awaiting transportation. Such tasks may take up a large part of a staff's time.} Under the balanced policy, there is less outflow of books from this branch overall, and the distribution is more even spatially. One aspect of the browser experience is how many of the most popular books are in the branch. In \Cref{fig:inventorychangeonebranch}, we compare the inventory levels of the 60 most popular book titles in High Bridge, under each of the historical and balanced policies. We observe that after warm-up, inventory under the balanced policy stabilizes at levels higher than the historical policy---even though they started at the same inventory level.

Although we are not specifically optimizing for these more nuanced measures, a balanced policy under our objectives naturally induces such properties. Furthermore, our framework is general and flexible enough to incorporate these measures into the objective, if they are of interest to library staff and system policymakers. In \Cref{app:sup}, we present additional information about these results, and discuss the relative contribution of our two policy levers.




\section{Discussion}

We design practical policies for holds systems in urban libraries, leveraging techniques from network revenue management and simulation optimization. First, we derive a $\frac{1}{2}$-approximation holds fulfillment method, that optimizes overall usage through the choice of source branch from which a holds request is fulfilled. Second, we embed this near-optimal fulfillment policy into a simulation-optimization framework, using which we can additionally choose a browser \textit{reservation} policy to jointly optimize overall usage and browser experience metrics. Finally, we apply our approach empirically using data from the New York Public Library, tracing a Pareto frontier of policies that optimize combinations of overall usage and browser experience. Evaluating one balanced policy in detail, we find that it improves upon both metrics over the historical policy and especially improves spatial and income-based equity, for both overall usage and browser experience. 

Methodologically, our work advances the data-driven optimization of implementable policies in hybrid computational-offline systems. We first derive an optimal policy (for one objective) theoretically, and then embed it into a simulation-optimization framework to optimize a second lever and two objectives, jointly. As a result, our approach provides both theoretical guarantees (for one objective) and flexibility to adapt to further constraints and desiderata from practice. We believe such an approach will be fruitful in other application domains, especially in characterizing tradeoffs between efficiency and another objective---substantial theoretical work already optimizes efficiency, and a characteristic of non-efficiency metrics (especially fairness-related ones) is that they may be challenging to theoretically analyze, given their global and non-linear nature. In this way, we hope that our work contributes to the growing literature on efficiency-equity tradeoffs in practice, both methodologically and by providing an empirical demonstration in a novel domain. 

One limitation of our current analysis is the  lack of a validation using out-of-sample data under distribution shift---our methods are optimized and then evaluated using the same simulator (though different sample path realizations). One possible approach is to validate policies on a simulator calibrated to other periods of historical data. We expect our approach to perform simularly when doing so---library usage patterns do not change dramatically yearly; as we need to consider a large number of book titles in the simulation (otherwise the simulator is not reflective of the overall usage pattern of each branch), this naturally smooths distribution shifts in usage patterns of individual book titles. Another possible approach is to optimize policies on one subset of book titles, and then validate on another subset. This approach risks losing the representativeness of the simulator to actual usage, when the policy would be applied to all books simultaneously. 


\paragraph{Future work} The best way to validate our results out-of-distribution is to put them into practice. As discussed, our policy solutions can be integrated into NYPL's current system without substantial computational changes, and we are actively discussing the feasibility of integrating these policy innovations into practice.
Our framework can also be easily applied to study this design question in other urban public library systems. For example, the public library systems in Los Angeles (72 branches), Chicago (81 branches), Houston (44 branches), and Seattle (27 branches) all operate similar holds systems, and have the capacity to carry out our analyses without significant modifications to our approach. Upon publication, we will share our code and simulator for such adaptation.  

Theoretically, our work also raises novel questions for the network revenue management literature. First, how to optimize both usage (fulfillment of demand using resources) and browser experience (value of resources on the shelves). This latter desiderata is related to but distinct from other objectives, such as the quality of assortments offered to customers (as considered in \citet{van2002variety}) and supply chain robustness. Second, how to jointly optimize fulfillment (and related) policies across a large set of \textit{heterogeneous} goods, when an offline policy must be shared across goods due to practical implementation constraints.



\newpage
\bibliographystyle{plainnat}
\bibliography{ref.bib}

\begin{thebibliography}{47}
\providecommand{\natexlab}[1]{#1}
\providecommand{\url}[1]{\texttt{#1}}
\expandafter\ifx\csname urlstyle\endcsname\relax
  \providecommand{\doi}[1]{doi: #1}\else
  \providecommand{\doi}{doi: \begingroup \urlstyle{rm}\Url}\fi

\bibitem[Azizi et~al.(2018)Azizi, Vayanos, Wilder, Rice, and Tambe]{azizi2018designing}
Mohammad~Javad Azizi, Phebe Vayanos, Bryan Wilder, Eric Rice, and Milind Tambe.
\newblock Designing fair, efficient, and interpretable policies for prioritizing homeless youth for housing resources.
\newblock In \emph{Integration of Constraint Programming, Artificial Intelligence, and Operations Research: 15th International Conference, CPAIOR 2018, Delft, The Netherlands, June 26--29, 2018, Proceedings 15}, pages 35--51. Springer, 2018.

\bibitem[Balandat et~al.(2020)Balandat, Karrer, Jiang, Daulton, Letham, Wilson, and Bakshy]{balandat2020botorch}
Maximilian Balandat, Brian Karrer, Daniel~R. Jiang, Samuel Daulton, Benjamin Letham, Andrew~Gordon Wilson, and Eytan Bakshy.
\newblock {BoTorch: A Framework for Efficient Monte-Carlo Bayesian Optimization}.
\newblock In \emph{Advances in Neural Information Processing Systems 33}, 2020.

\bibitem[Banerjee and Smilowitz(2019)]{banerjee2019incorporating}
Dipayan Banerjee and Karen Smilowitz.
\newblock Incorporating equity into the school bus scheduling problem.
\newblock \emph{Transportation research part E: logistics and transportation review}, 131:\penalty0 228--246, 2019.

\bibitem[Banerjee et~al.(2023)Banerjee, Hssaine, and Sinclair]{banerjee2023online}
Siddhartha Banerjee, Chamsi Hssaine, and Sean~R Sinclair.
\newblock Online fair allocation of perishable resources.
\newblock \emph{ACM SIGMETRICS Performance Evaluation Review}, 51\penalty0 (1):\penalty0 55--56, 2023.

\bibitem[Bertsimas et~al.(2011)Bertsimas, Farias, and Trichakis]{bertsimas2011price}
Dimitris Bertsimas, Vivek~F Farias, and Nikolaos Trichakis.
\newblock The price of fairness.
\newblock \emph{Operations research}, 59\penalty0 (1):\penalty0 17--31, 2011.

\bibitem[Bertsimas et~al.(2012)Bertsimas, Farias, and Trichakis]{bertsimas2012efficiency}
Dimitris Bertsimas, Vivek~F Farias, and Nikolaos Trichakis.
\newblock On the efficiency-fairness trade-off.
\newblock \emph{Management Science}, 58\penalty0 (12):\penalty0 2234--2250, 2012.

\bibitem[Bertsimas et~al.(2020)Bertsimas, Delarue, Eger, Hanlon, and Martin]{bertsimas2020bus}
Dimitris Bertsimas, Arthur Delarue, William Eger, John Hanlon, and Sebastien Martin.
\newblock Bus routing optimization helps boston public schools design better policies.
\newblock \emph{INFORMS Journal on Applied Analytics}, 50\penalty0 (1):\penalty0 37--49, 2020.

\bibitem[Besbes and Zeevi(2012)]{besbes2012blind}
Omar Besbes and Assaf Zeevi.
\newblock Blind network revenue management.
\newblock \emph{Operations research}, 60\penalty0 (6):\penalty0 1537--1550, 2012.

\bibitem[Chen et~al.(2023)Chen, Liang, Golrezaei, and Bouneffouf]{chen2023interpolating}
Qinyi Chen, Jason Cheuk~Nam Liang, Negin Golrezaei, and Djallel Bouneffouf.
\newblock Interpolating item and user fairness in multi-sided recommendations.
\newblock \emph{Available at SSRN 4476512}, 2023.

\bibitem[Cheng et~al.(2021)Cheng, Wu, Moen, and Hong]{cheng2021assessing}
Wenting Cheng, Jiahui Wu, William Moen, and Lingzi Hong.
\newblock Assessing the spatial accessibility and spatial equity of public libraries' physical locations.
\newblock \emph{Library \& Information Science Research}, 43\penalty0 (2):\penalty0 101089, 2021.

\bibitem[Chung et~al.(2018)Chung, Freund, and Shmoys]{chung2018bike}
Hangil Chung, Daniel Freund, and David~B Shmoys.
\newblock Bike angels: An analysis of citi bike's incentive program.
\newblock In \emph{Proceedings of the 1st ACM SIGCAS Conference on Computing and Sustainable Societies}, pages 1--9, 2018.

\bibitem[Daulton et~al.(2021)Daulton, Balandat, and Bakshy]{daulton2021parallel}
Samuel Daulton, Maximilian Balandat, and Eytan Bakshy.
\newblock Parallel bayesian optimization of multiple noisy objectives with expected hypervolume improvement.
\newblock \emph{Advances in Neural Information Processing Systems}, 34:\penalty0 2187--2200, 2021.

\bibitem[Del Barrio-Tellado et~al.(2021)Del Barrio-Tellado, G{\'o}mez-Vega, G{\'o}mez-Zapata, and Herrero-Prieto]{del2021urban}
Mar{\'\i}a~Jos{\'e} Del Barrio-Tellado, Mafalda G{\'o}mez-Vega, Jonathan~Daniel G{\'o}mez-Zapata, and Luis~C{\'e}sar Herrero-Prieto.
\newblock Urban public libraries: performance analysis using dynamic-network-dea.
\newblock \emph{Socio-Economic Planning Sciences}, 74:\penalty0 100928, 2021.

\bibitem[Delarue et~al.(2024)Delarue, Lian, and Martin]{delarue2024algorithmic}
Arthur Delarue, Zhen Lian, and Sebastien Martin.
\newblock Algorithmic precision and human decision: A study of interactive optimization for school schedules.
\newblock \emph{Available at SSRN 4324076}, 2024.

\bibitem[Ferreira et~al.(2018)Ferreira, Simchi-Levi, and Wang]{ferreira2018online}
Kris~Johnson Ferreira, David Simchi-Levi, and He~Wang.
\newblock Online network revenue management using thompson sampling.
\newblock \emph{Operations research}, 66\penalty0 (6):\penalty0 1586--1602, 2018.

\bibitem[Freund et~al.(2020)Freund, Norouzi-Fard, Paul, Wang, Henderson, and Shmoys]{freund2020data}
Daniel Freund, Ashkan Norouzi-Fard, Alice Paul, Carter Wang, Shane~G Henderson, and David~B Shmoys.
\newblock Data-driven rebalancing methods for bike-share systems.
\newblock In \emph{Analytics for the sharing economy: Mathematics, engineering and business perspectives}, pages 255--278. Springer, 2020.

\bibitem[Gallego and Phillips(2004)]{gallego2004revenue}
Guillermo Gallego and Robert Phillips.
\newblock Revenue management of flexible products.
\newblock \emph{Manufacturing \& Service Operations Management}, 6\penalty0 (4):\penalty0 321--337, 2004.

\bibitem[Gallego and Van~Ryzin(1994)]{gallego1994optimal}
Guillermo Gallego and Garrett Van~Ryzin.
\newblock Optimal dynamic pricing of inventories with stochastic demand over finite horizons.
\newblock \emph{Management science}, 40\penalty0 (8):\penalty0 999--1020, 1994.

\bibitem[Gallego et~al.(2004)Gallego, Iyengar, Phillips, and Dubey]{gallego2004managing}
Guillermo Gallego, Garud Iyengar, Robert Phillips, and Abhay Dubey.
\newblock Managing flexible products on a network.
\newblock \emph{Available at SSRN 3567371}, 2004.

\bibitem[Greenwood et~al.(2024)Greenwood, Chiniah, and Garg]{greenwood2024useritem}
Sophie Greenwood, Sudalakshmee Chiniah, and Nikhil Garg.
\newblock User-item fairness tradeoffs in recommendations.
\newblock In \emph{The Thirty-eighth Annual Conference on Neural Information Processing Systems}, 2024.
\newblock URL \url{https://openreview.net/forum?id=ZOZjMs3JTs}.

\bibitem[Guo et~al.(2017)Guo, Chan, and Yip]{guo2017spatial}
Yingqi Guo, Chee~Hon Chan, and Paul~SF Yip.
\newblock Spatial variation in accessibility of libraries in hong kong.
\newblock \emph{Library \& Information Science Research}, 39\penalty0 (4):\penalty0 319--329, 2017.

\bibitem[He et~al.(2021)He, Zheng, Belavina, and Girotra]{he2021customer}
Pu~He, Fanyin Zheng, Elena Belavina, and Karan Girotra.
\newblock Customer preference and station network in the london bike-share system.
\newblock \emph{Management Science}, 67\penalty0 (3):\penalty0 1392--1412, 2021.

\bibitem[Jalota et~al.(2021)Jalota, Solovey, Gopalakrishnan, Zoepf, Balakrishnan, and Pavone]{jalota2021efficiency}
Devansh Jalota, Kiril Solovey, Karthik Gopalakrishnan, Stephen Zoepf, Hamsa Balakrishnan, and Marco Pavone.
\newblock When efficiency meets equity in congestion pricing and revenue refunding schemes.
\newblock In \emph{Proceedings of the 1st ACM Conference on Equity and Access in Algorithms, Mechanisms, and Optimization}, pages 1--11, 2021.

\bibitem[Jian et~al.(2016)Jian, Freund, Wiberg, and Henderson]{jian2016simulation}
Nanjing Jian, Daniel Freund, Holly~M Wiberg, and Shane~G Henderson.
\newblock Simulation optimization for a large-scale bike-sharing system.
\newblock In \emph{2016 Winter Simulation Conference (WSC)}, pages 602--613. IEEE, 2016.

\bibitem[Koenecke et~al.(2023)Koenecke, Giannella, Willer, and Goel]{koenecke2023popular}
Allison Koenecke, Eric Giannella, Robb Willer, and Sharad Goel.
\newblock Popular support for balancing equity and efficiency in resource allocation: A case study in online advertising to increase welfare program awareness.
\newblock In \emph{Proceedings of the International AAAI Conference on Web and Social Media}, volume~17, pages 494--506, 2023.

\bibitem[Liang et~al.(2022)Liang, Lu, and Mu]{liang2022algorithmic}
Annie Liang, Jay Lu, and Xiaosheng Mu.
\newblock Algorithmic design: Fairness versus accuracy.
\newblock In \emph{Proceedings of the 23rd ACM Conference on Economics and Computation}, pages 58--59, 2022.

\bibitem[Liu and Garg(2024)]{liu2024redesigning}
Zhi Liu and Nikhil Garg.
\newblock Redesigning service level agreements: Equity and efficiency in city government operations.
\newblock In \emph{Proceedings of the 25th ACM Conference on Economics and Computation}, pages 309--309, 2024.

\bibitem[Liu et~al.(2024)Liu, Rankin, and Garg]{liu2024identifying}
Zhi Liu, Sarah Rankin, and Nikhil Garg.
\newblock Identifying and addressing disparities in public libraries with bayesian latent variable modeling.
\newblock In \emph{Proceedings of the AAAI Conference on Artificial Intelligence}, volume~38, pages 22258--22265, 2024.

\bibitem[Ma et~al.(2020)Ma, Rusmevichientong, Sumida, and Topaloglu]{ma2020approximation}
Yuhang Ma, Paat Rusmevichientong, Mika Sumida, and Huseyin Topaloglu.
\newblock An approximation algorithm for network revenue management under nonstationary arrivals.
\newblock \emph{Operations Research}, 68\penalty0 (3):\penalty0 834--855, 2020.

\bibitem[Maheshwari et~al.(2024)Maheshwari, Kulkarni, Pai, Yang, Wu, and Sastry]{maheshwari2024congestion}
Chinmay Maheshwari, Kshitij Kulkarni, Druv Pai, Jiarui Yang, Manxi Wu, and Shankar Sastry.
\newblock Congestion pricing for efficiency and equity: Theory and applications to the san francisco bay area.
\newblock \emph{arXiv preprint arXiv:2401.16844}, 2024.

\bibitem[Manshadi et~al.(2021)Manshadi, Niazadeh, and Rodilitz]{manshadi2021fair}
Vahideh Manshadi, Rad Niazadeh, and Scott Rodilitz.
\newblock Fair dynamic rationing.
\newblock In \emph{Proceedings of the 22nd ACM Conference on Economics and Computation}, pages 694--695, 2021.

\bibitem[Manshadi et~al.(2023)Manshadi, Rodilitz, Saban, and Suresh]{manshadi2023redesigning}
Vahideh Manshadi, Scott Rodilitz, Daniela Saban, and Akshaya Suresh.
\newblock Redesigning volunteermatch's ranking algorithm: Toward more equitable access to volunteers.
\newblock \emph{Available at SSRN 4497747}, 2023.

\bibitem[Meier et~al.(2023)Meier, Davis, and Xu]{Meier2023Effectiveness}
K.~Meier, Jourdan~A. Davis, and Xiaoyang Xu.
\newblock Effectiveness, efficiency and equity tradeoffs in public programs: A citizen experiment.
\newblock \emph{Public Administration Review}, 2023.
\newblock \doi{10.1111/puar.13690}.

\bibitem[Monachou and Stoica(2022)]{monachou2022fairness}
Faidra Monachou and Ana-Andreea Stoica.
\newblock Fairness and equity in resource allocation and decision-making: an annotated reading list.
\newblock \emph{ACM SIGecom Exchanges}, 20\penalty0 (1):\penalty0 64--66, 2022.

\bibitem[{NYC Mayor's Office of Operations}(2023)]{nypl2023}
{NYC Mayor's Office of Operations}.
\newblock {Mayor's Management Report}, 2023.
\newblock URL \url{https://www.nyc.gov/site/operations/performance/mmr.page}.

\bibitem[Ostrovsky and Yang(2024)]{ostrovsky2024effective}
Michael Ostrovsky and Frank Yang.
\newblock Effective and equitable congestion pricing: New york city and beyond.
\newblock 2024.

\bibitem[Patro et~al.(2022)Patro, Porcaro, Mitchell, Zhang, Zehlike, and Garg]{patroranking}
Gourab~K. Patro, Lorenzo Porcaro, Laura Mitchell, Qiuyue Zhang, Meike Zehlike, and Nikhil Garg.
\newblock Fair ranking: a critical review, challenges, and future directions.
\newblock In \emph{Proceedings of the 2022 ACM Conference on Fairness, Accountability, and Transparency}, FAccT '22, page 1929–1942, New York, NY, USA, 2022. Association for Computing Machinery.
\newblock ISBN 9781450393522.
\newblock \doi{10.1145/3531146.3533238}.
\newblock URL \url{https://doi.org/10.1145/3531146.3533238}.

\bibitem[Real et~al.(2014)Real, Bertot, and Jaeger]{real2014rural}
Brian Real, John~Carlo Bertot, and Paul~T Jaeger.
\newblock Rural public libraries and digital inclusion: Issues and challenges.
\newblock \emph{Information Technology and Libraries}, 33\penalty0 (1):\penalty0 6--24, 2014.

\bibitem[Reichmann(2004)]{reichmann2004measuring}
Gerhard Reichmann.
\newblock Measuring university library efficiency using data envelopment analysis.
\newblock 2004.

\bibitem[Shim(2003)]{shim2003applying}
Wonsik Shim.
\newblock Applying dea technique to library evaluation in academic research libraries.
\newblock 2003.

\bibitem[Sin(2011)]{sin2011neighborhood}
Sei-Ching~Joanna Sin.
\newblock Neighborhood disparities in access to information resources: Measuring and mapping us public libraries’ funding and service landscapes.
\newblock \emph{Library \& Information Science Research}, 33\penalty0 (1):\penalty0 41--53, 2011.

\bibitem[Sinclair et~al.(2022)Sinclair, Banerjee, and Yu]{sinclair2022sequential}
Sean~R Sinclair, Siddhartha Banerjee, and Christina~Lee Yu.
\newblock Sequential fair allocation: Achieving the optimal envy-efficiency tradeoff curve.
\newblock \emph{ACM SIGMETRICS Performance Evaluation Review}, 50\penalty0 (1):\penalty0 95--96, 2022.

\bibitem[Singla et~al.(2015)Singla, Santoni, Bartók, Mukerji, Meenen, and Krause]{Singla_Santoni_Bartók_Mukerji_Meenen_Krause_2015}
Adish Singla, Marco Santoni, Gábor Bartók, Pratik Mukerji, Moritz Meenen, and Andreas Krause.
\newblock Incentivizing users for balancing bike sharing systems.
\newblock \emph{Proceedings of the AAAI Conference on Artificial Intelligence}, 29\penalty0 (1), Feb. 2015.
\newblock \doi{10.1609/aaai.v29i1.9251}.
\newblock URL \url{https://ojs.aaai.org/index.php/AAAI/article/view/9251}.

\bibitem[Talluri and Van~Ryzin(1998)]{talluri1998analysis}
Kalyan Talluri and Garrett Van~Ryzin.
\newblock An analysis of bid-price controls for network revenue management.
\newblock \emph{Management science}, 44\penalty0 (11-part-1):\penalty0 1577--1593, 1998.

\bibitem[Torrico et~al.(2024)Torrico, Boonsiriphatthanajaroen, Garg, Lodi, and Mainguy]{torrico2024equitable}
Alfredo Torrico, Natthawut Boonsiriphatthanajaroen, Nikhil Garg, Andrea Lodi, and Hugo Mainguy.
\newblock Equitable congestion pricing under the markovian traffic model: An application to bogota.
\newblock \emph{ACM Conference on Economics and Computation}, 2024.

\bibitem[Van~Herpen and Pieters(2002)]{van2002variety}
Erica Van~Herpen and Rik Pieters.
\newblock The variety of an assortment: An extension to the attribute-based approach.
\newblock \emph{Marketing Science}, 21\penalty0 (3):\penalty0 331--341, 2002.

\bibitem[Zhu and Topaloglu(2024)]{zhu2024performance}
Wenchang Zhu and Huseyin Topaloglu.
\newblock Performance guarantees for network revenue management with flexible products.
\newblock \emph{Manufacturing \& Service Operations Management}, 26\penalty0 (1):\penalty0 252--270, 2024.

\end{thebibliography}

\newpage
\pagebreak
\appendix

\section{Omitted proofs}
\label{app:proofs}

\subsection{Proof for \Cref{thm:main}}

The proof for our main theorem relies on the following lemma, which states that the value function approximation $H_1(\textbf{c})$ underestimates $V_1(\textbf{c})$ by at most a factor of 2:

\begin{restatable}{lemma}{propvaluefunc}
    \label{lem:1}
     $2H_1(\textbf{c}) \ge V_1(\textbf{c})$.
\end{restatable}



\begin{proof}

We first note that from \Cref{alg:one}, there are two observations:
\begin{equation}
\hat{\gamma}_{at} - \hat{\gamma}_{a, t+1} = \sum_{j\in P}\frac{\lambda_{jt}}{c_a} \mathbf{1}(a = a^*_{jt})(r_j-\hat{\gamma}_{a, t+1})^+ \ge 0 \quad  \Rightarrow \quad  \hat{\gamma}_{a1} \ge \hat{\gamma}_{a2}\ge \dots \ge \hat{\gamma}_{a,T+1}=0, \label{eq:gammanonneg}
\end{equation}
and
\begin{equation}
    \sum_{i\in B} \left(\hat{\gamma}_{at} - \hat{\gamma}_{a, t+1}\right) = \sum_{j \in P} \frac{\lambda_j}{c_a} \max_{a\in A_j} (1-\gamma_{a, t+1})^+. \label{eq:gammareform}
\end{equation}

Consider the dual of (LP) as follows.
\begin{subequations}
    \begin{align}
        \text{LP-Dual} = \min_{\mathbf{\mu, \sigma}}\quad &\sum_{a\in A}c_a \mu_a + T\sum_{j\in P}\sigma_j \lambda_{j}\\
        \text{s.t.}\quad  &\mu_a + \sigma_j \ge r_j, & \forall j\in P, a\in A_j \label{constr:musigma}\\
        &\mu_a,\sigma_j \ge 0, &\forall a\in A, j\in P.
    \end{align}
\end{subequations}

Note that constraint \ref{constr:musigma} together with non-negativity is equivalent to 
$$
\sigma_j \ge \max_{i\in A_j} (r_j-\mu_i)^+, \quad \forall j\in P.
$$

Substituting this into the problem, we have the following equivalent form of (LP-Dual):
\begin{equation}
    \text{LP-Dual} = \min_{\mu\ge 0} \ \left[ \sum_{a\in A} c_a \mu_a + T\sum_{j\in P}\lambda_{j}\max_{a\in A_j}(r_j-\mu_a)^+ \right]
\end{equation}

Based on \Cref{eq:gammanonneg}, $\mu_a = \hat{\gamma}_{a1}\ge 0, \forall a\in A$ is a feasible solution to (LP-Dual), and since (LP-Dual) is a minimization problem, substituting this set of feasible solution and using weak duality we get
\begin{subequations}
    \begin{align}
        \text{LP} &\le \sum_{a\in A} c_a \hat{\gamma}_{a1} + T \sum_{j\in P} \lambda_j \max_{a\in A_j}(r_j-\hat{\gamma}_{a1})^+ \label{eq:lem1a}\\
        & \le \sum_{a\in A} c_a \hat{\gamma}_{a1} + \sum_{t \in [T]} \sum_{j\in P} \lambda_j \max_{a\in A_j}(r_j-\hat{\gamma}_{at})^+ \label{eq:lem1b}\\
        &= \sum_{a\in A} c_a \hat{\gamma}_{a1} + \sum_{t\in [T]} \sum_{a\in A} c_a (\hat{\gamma}_{at}- \hat{\gamma}_{a,t+1})\label{eq:lem1c}\\
        & = 2 \sum_{a\in A} c_a \hat{\gamma}_{a1} = 2H_1(\textbf{c}),\label{eq:lem1d}
    \end{align}
\end{subequations}
where from \Cref{eq:lem1a} to \Cref{eq:lem1b} we use \Cref{eq:gammanonneg}, and use \Cref{eq:gammareform} to get to \Cref{eq:lem1c}. The final step follows from $\sum_{t\in [T]}(\hat{\gamma}_{at}- \hat{\gamma}_{a,t+1}) = \hat{\gamma}_{a1} - \hat{\gamma}_{a,T+1} = \hat{\gamma}_{a1}$.

Since by \Cref{prop:LPUB}, $\text{LP}\ge V_1(\textbf{c})$, we get $2H_1(\textbf{c})\ge LP \ge V_1(\textbf{c})$.
    
\end{proof}

We are now ready to show the main result.

\mainthm*

\begin{proof}
    We show $R_1(\textbf{c})\ge H_1(\textbf{c})$ through induction. Since $R_{T+1}(\textbf{x}) = H_{T+1}(\textbf{x}) = 0, \forall \textbf{x}$, we aim to show that:
    $$
    R_{t+1}(\textbf{x}) \ge H_{t+1}(\textbf{x}), \forall \textbf{x}, \quad  \Rightarrow  \quad J_{t}(\textbf{x}) \ge H_{t}(\textbf{x}), \forall \textbf{x}.
    $$

    To see this, we have the following:
    \begin{subequations}
        \begin{align}
            R_{t}(\textbf{x}) & = \sum_{j\in P} \lambda_j \sum_{a\in A_j} U_{ja}^t(\textbf{x})(r_j + R_{t+1}(\textbf{x} - \textbf{e}_a)) + \left(1-\sum_{j\in P}\lambda_j \sum_{a\in A_j}U_{ja}^t(\textbf{x})\right)R_{t+1}(\textbf{x})\\
            &\ge \sum_{j\in P} \lambda_j \sum_{a\in A_j} U_{ja}^t(\textbf{x})(r_j + H_{t+1}(\textbf{x} - \textbf{e}_a)) + \left(1-\sum_{j\in P}\lambda_j \sum_{a\in A_j}U_{ja}^t(\textbf{x})\right)H_{t+1}(\textbf{x})\label{thmb}\\
            & = H_{t+1}(\textbf{x}) + \sum_{j\in P} \lambda_j \sum_{a\in A_j} U_{ja}^t(\textbf{x})\left(r_j + H_{t+1}(\textbf{x}-\textbf{e}_a) - H_{t+1}(\textbf{x})\right)\label{thmc}\\
            & = H_{t+1}(\textbf{x}) + \sum_{j\in P} \lambda_j \mathbf{1}\left(\sum_{a\in A_j}x_a>0\right)\left(r_j +  \max_{a\in A_j: x_a>0}H_{t+1}(\textbf{x}-\textbf{e}_a) - H_{t+1}(\textbf{x})\right)^+\label{thmd}\\
            & \ge H_{t+1}(\textbf{x}) + \sum_{j\in P} \lambda_j \mathbf{1}\left(x_{a^*_{jt}}>0\right)\left(r_j +  H_{t+1}(\textbf{x}-\textbf{e}_{a^*_{jt}}) - H_{t+1}(\textbf{x})\right)^+\label{thme}\\
            & = H_{t+1}(\textbf{x}) + \sum_{j\in P} \lambda_j \mathbf{1}\left(x_{a^*_{jt}}>0\right)\left(r_j -\hat{\gamma}_{a^*_{jt}, t+1}\right)^+\label{thmf}\\
            & \ge H_{t+1}(\textbf{x}) + \sum_{j\in P} \lambda_j \frac{x_{a_{jt}^*}}{c_{a_{jt}^*}}\left(r_j -\hat{\gamma}_{a^*_{jt}, t+1}\right)^+\label{thmg}\\
            & = \sum_{a\in A} \hat{\gamma}_{a, t+1} x_a + \left(\hat{\gamma}_{a_{jt}^*, t} - \hat{\gamma}_{a_{jt}^*, t+1}\right)x_{a_{jt}^*}\label{thmh}\\
            & = \sum_{a\in A} \hat{\gamma}_{a, t+1} x_a + \sum_{a\in A}\left(\hat{\gamma}_{a, t} - \hat{\gamma}_{a, t+1}\right)x_{a}\label{thmi}\\
            & = \sum_{a\in A} \hat{\gamma}_{a, t} x_a\label{thmj}\\
            & = H_t(\textbf{x}). \label{thmk}
        \end{align}
    \end{subequations}

\Cref{thmb}: From induction hypothesis $R_{t+1}(\textbf{x}) \ge H_{t+1}(\textbf{x})$, and that the coefficients are positive.

\Cref{thmc}: Rearranging terms.

\Cref{thmd}: By definition of $U_{ji}^t(\textbf{x})$.

\Cref{thme}: Since $a^*_{jt} \in A_j$, we have $\mathbf{1}\left(\sum_{i\in A_j}x_i>0\right) \ge \mathbf{1}\left(x_{a^*_{jt}}>0\right)$, and $\max_{i\in A_j: x_i>0}H_{t+1}(\textbf{x}-\textbf{e}_i)\ge H_{t+1}(\textbf{x}-\textbf{e}_{a^*_{jt}})$.

\Cref{thmf}: By definition of the approximate value function $H_t(\textbf{x})$.

\Cref{thmg}: $\mathbf{1}\left(x_{a^*_{jt}}>0\right) \ge \frac{x_{a_{jt}^*}}{c_{a_{jt}^*}}$.

\Cref{thmh}: By definition of $H_t(\textbf{x})$ and the update rule of $\hat{\gamma}$.

\Cref{thmi}: By the update rule, only $a^*_{jt}$ has its corresponding $\hat{\gamma}$ updated.

\Cref{thmj} and \ref{thmk}: by rearranging terms and definition of $H_t(\textbf{x})$.

Thus we have shown by induction $J_1(\textbf{c})\ge H_1(\textbf{c})$, and by \Cref{lem:1}, we get the desired result.
    
\end{proof}

\section{Supplementary information on empirical application}
\label{app:backgroun_application}

In this section, we provide more information on the inputs to our empirical application and discuss some potential biases.

The data used by our simulator consists of patron arrival probabilities and starting inventory levels. Patron arrival probabilities are calculated using latent demand sizes at each branch shown in \Cref{fig:latentdemandsize}, hold usage fractions shown in \Cref{fig:1b}, and latent desirability of book titles shown in \Cref{fig:latentdes}. The number of book titles and total copies of book titles available at each branch at the start of the simulation are shown in \Cref{fig:inventory}.

We consider a total of 3,809 book titles, with a total of 131,476 copies across 84 branches. These 3,809 book titles are the ones that are available at 20 or more branches. We conduct this filtering for two reasons. First, if a book title is only available at a few branches, including it in our simulation risks biasing our results with outlying usage patterns: e.g., a particular dictionary may be only available at a few larger branches, and are often transported around by requests. Such usage does not contribute significantly to and could introduce bias to the overall browser reserve and usage metrics we are primarily interested. Second, including these book titles efficiently captures overall usage with relatively low computational needs: as we point out in the main text, these titles represent 1.27\% of all book titles but 31.01\% of usage.

As bound by data privacy, we are unable to provide the exact data inputs to the simulation. Nevertheless, we provide the simulator itself and the Bayesian optimization procedure, implemented in Python, as code submission.

\begin{figure}[tbh]
    \centering
    
    \begin{subfigure}[b]{0.45\textwidth}
        \centering
        \includegraphics[width=\textwidth]{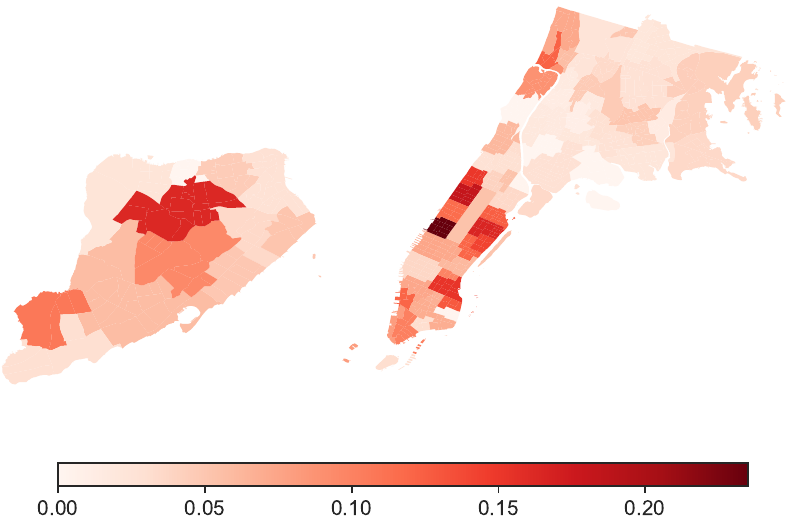}
        \caption{Latent patron demand sizes of each branch}
        \label{fig:latentdemandsize}
    \end{subfigure}
    \begin{subfigure}[b]{0.45\textwidth}
        \centering
        \includegraphics[width=\textwidth]{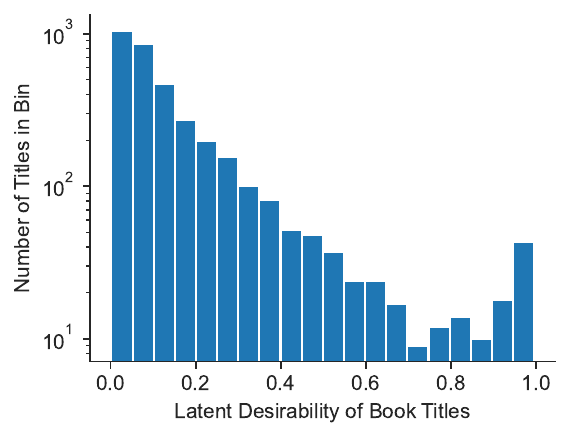}
        \caption{Histogram of latent desirability of book titles}
        \label{fig:latentdes}
    \end{subfigure}
    \hfill
    \caption{Latent demand sizes of each branch and latent desirability of each book title are used in conjunction with hold usage fraction of each branch (\Cref{fig:1b}) to calculate arrival probability of each type of patron. The calculation is presented in \Cref{eq:browserarrival}, and further adjustments are made for better calibration.}
    \label{fig:demandsizeanddes}
\end{figure}

\begin{figure}[tbh]
    \centering
    
    \begin{subfigure}[b]{0.45\textwidth}
        \centering
        \includegraphics[width=\textwidth]{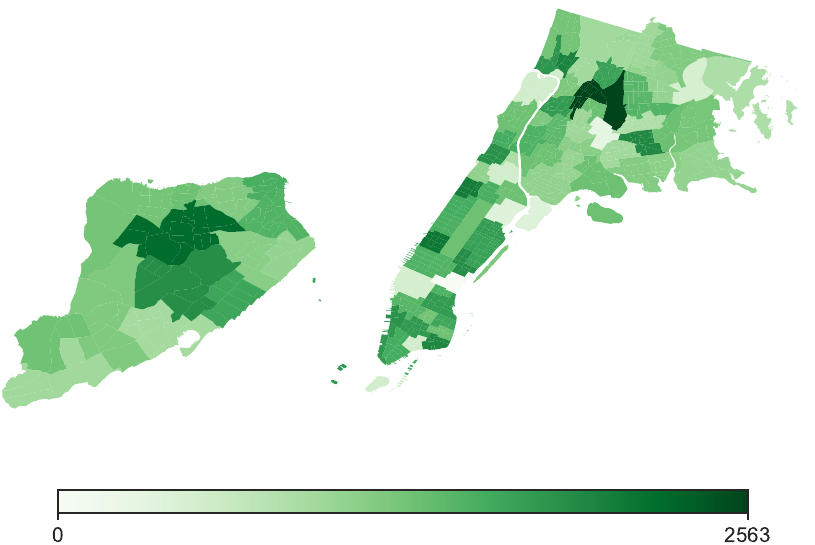}
        \caption{Number of available book titles}
        \label{fig:titles}
    \end{subfigure}
    \begin{subfigure}[b]{0.45\textwidth}
        \centering
        \includegraphics[width=\textwidth]{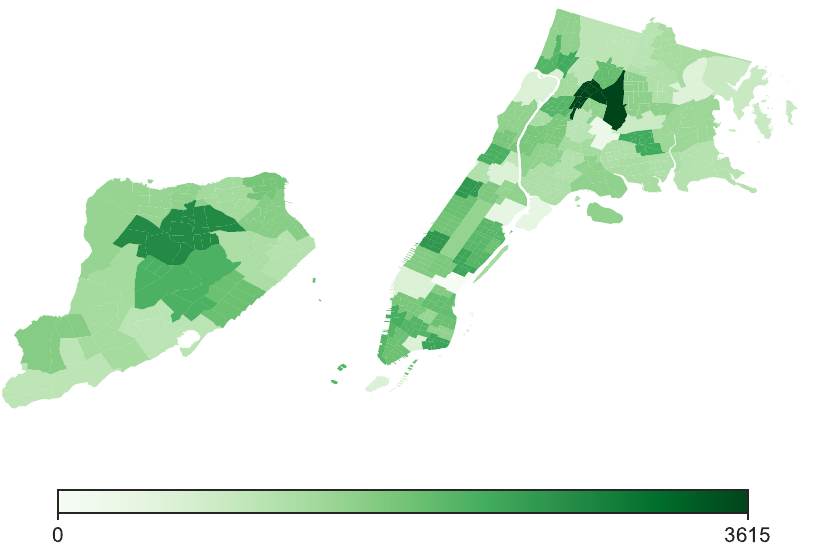}
        \caption{Number of total available copies}
        \label{fig:copies}
    \end{subfigure}
    \hfill
    \caption{Number of book titles and total copies available at each branch at the start of the simulation. We consider a total of 3,809 book titles, with a total of 131,476 copies across 84 branches.}
    \label{fig:inventory}
\end{figure}

\section{Supplementary results}
\label{app:sup}

\subsection{Additional information on policies presented in the main text}

In this subsection, we present some additional information on policies presented in the main text. \Cref{fig:implications_app} shows usage and browser experience ratios associated with three scenarios: historical data, balanced policy, and another policy that prioritizes browser experience. The first two scenarios are the same with \Cref{fig:implications}. We observe that the balanced policy fits in between these two relatively extreme cases. \Cref{fig:policy_params} presents the exact parameters of the balanced policy, consisting of a set of branch reserve fractions and priority tiers for each branch. \Cref{fig:netinflowchangeapp} provides additional information to \Cref{fig:netinflowchangeonebranch} and \Cref{fig:1c}: the balanced policy generally smooths net inflow for branches, with higher-income branches experiencing a decrease and lower-income branches experiencing an increase. 

\begin{figure}[tbh]
    \centering
    
    \includegraphics[width=\textwidth]{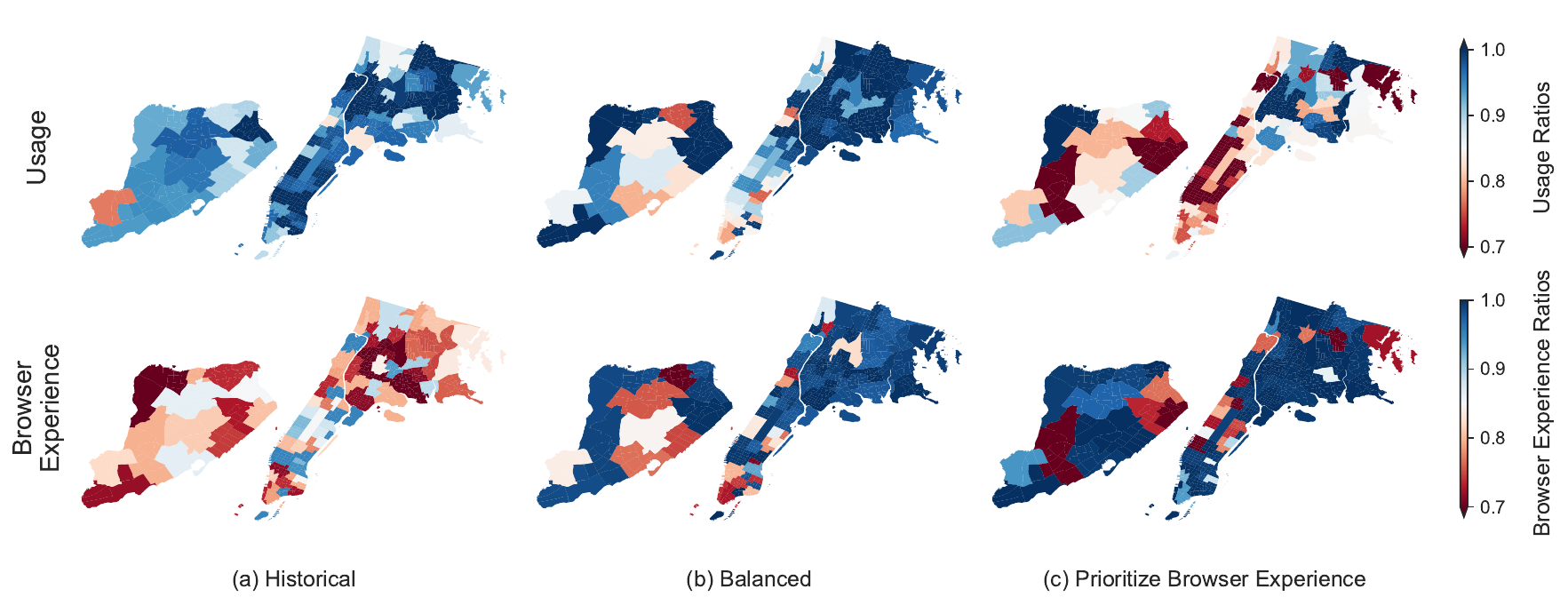}
    \caption{In comparison to \Cref{fig:implications_app}, we additionally include a policy that prioritizes browser experience even more, at a higher loss of usage. The balanced policy generally fits between historical data and a more extreme prioritization of browser experience, achieving a balance of the two objectives.}
    \label{fig:implications_app}
\end{figure}

\begin{figure}[tbh]
    \centering
    
    \begin{subfigure}[b]{0.47\textwidth}
        \centering
        \includegraphics[width=\textwidth]{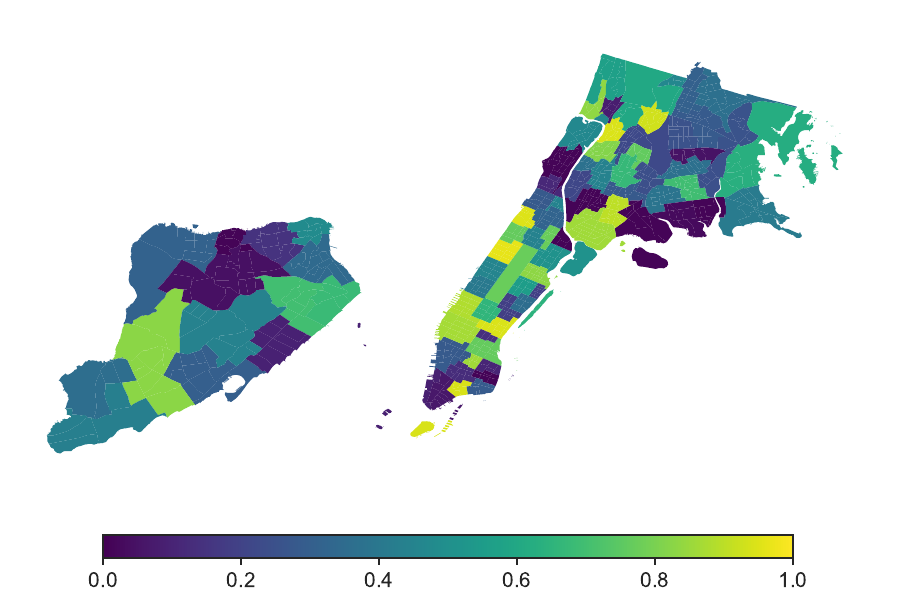}
        \caption{Branch reserve fractions}
        \label{fig:browser_reserve_map}
    \end{subfigure}
    \hfill
    \begin{subfigure}[b]{0.47\textwidth}
        \centering
        \includegraphics[width=\textwidth]{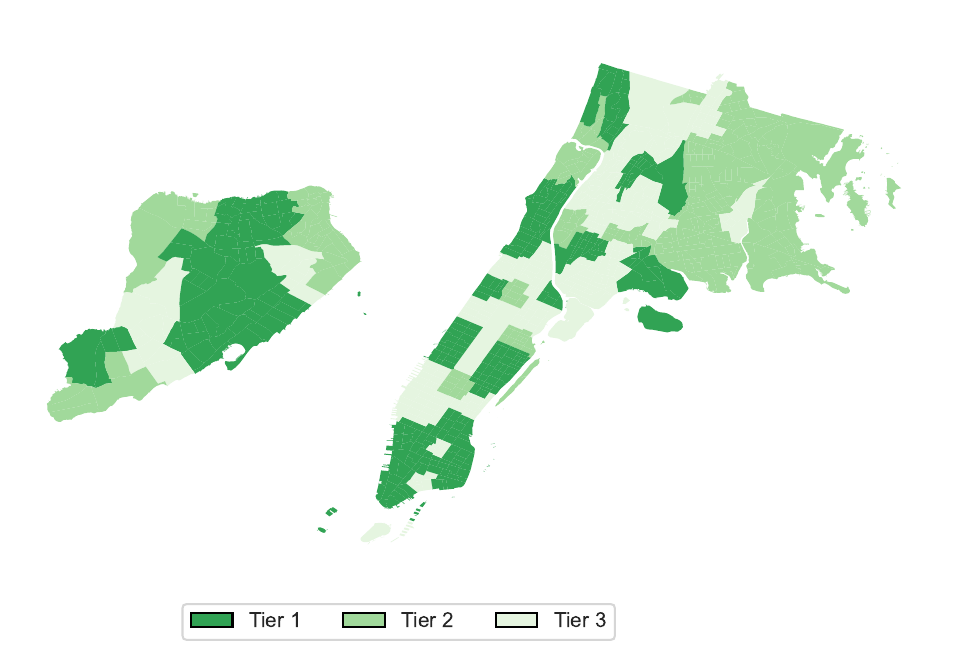}
        \caption{Priority tiers}
        \label{fig:implementable_fulfillment}
    \end{subfigure}

    \caption{Policy parameters of the balanced implementable policy. In general, higher branch reserves and lower priority tiers are associated with higher browser experience, whereas usage is a more complex function of the parameters.}
    \label{fig:policy_params}
\end{figure}

    

\begin{figure}[tbh]
    \centering
    \begin{subfigure}[b]{0.48\textwidth}
        \centering
        \includegraphics[width=\textwidth]{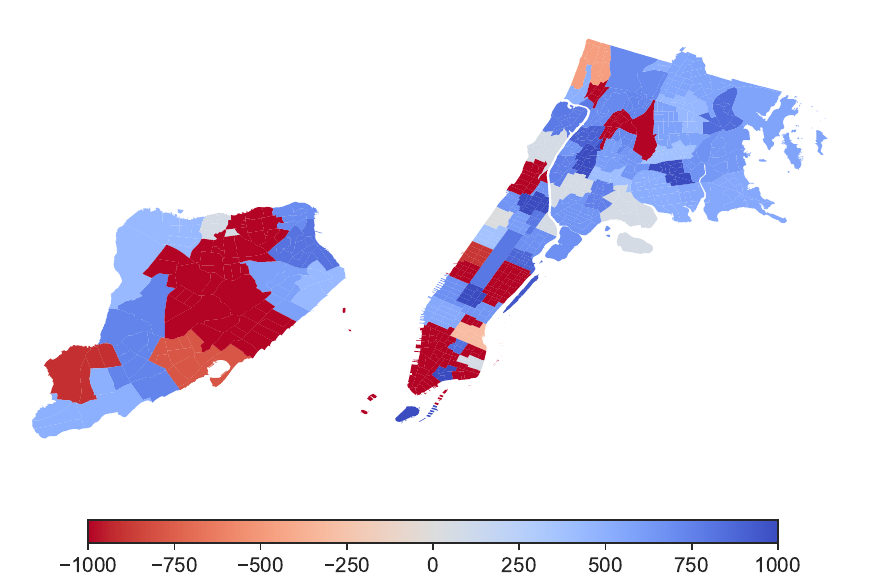}
        \caption{Change of net desirability inflow}
        \label{fig:collectionqualityratiochange}
    \end{subfigure}
    \begin{subfigure}[b]{0.48\textwidth}
        \centering
        \includegraphics[width=\textwidth]{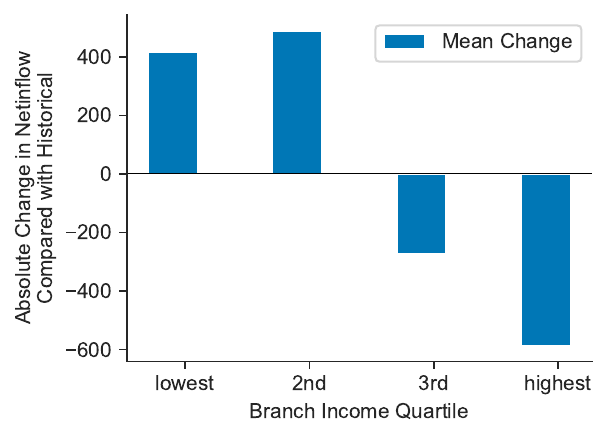}
        \caption{Absolute change for each income quartile}
        \label{fig:netinflowchange}
    \end{subfigure}
    \hfill
    \caption{Absolute change net desirability inflow, under the balanced implementable policy. Comparing (a) with \Cref{fig:1c}, the balanced policy alleviates disparities in net inflow. This is done by predominantly increasing inflow to lower-income branches, and decreasing inflow to higher income branches, as shown in (b).}
    \label{fig:netinflowchangeapp}
\end{figure}

\subsection{Empirical analyses with alternative browser experience objective}
\label{sec:nash_app}

\Cref{fig:pareto_nash_app} presents the Pareto frontier found when the browser experience objective is defined as the Nash welfare of all branches. Concretely:
$$
g^{\text{Nash}}(\vec{\beta}, \mathcal{F}) = \left(\prod_{i\in B} \frac{\CQ_i(\vec{\beta}, \mathcal{F})}{\CQ_i(\textbf{1}, \mathcal{F}^*)}\right)^{\frac{1}{n}},
$$
where $n$ is the number of branches, this definition aligns more with popular definitions of fairness and equity in multi-player settings. Our main insights remain the same under such a setting: though there is a substantial trade-off between usage and browser experience, there exists a large array of implementable policies that Pareto dominates historical practice.

\begin{figure}
    \centering
    \includegraphics[width=0.53\linewidth]{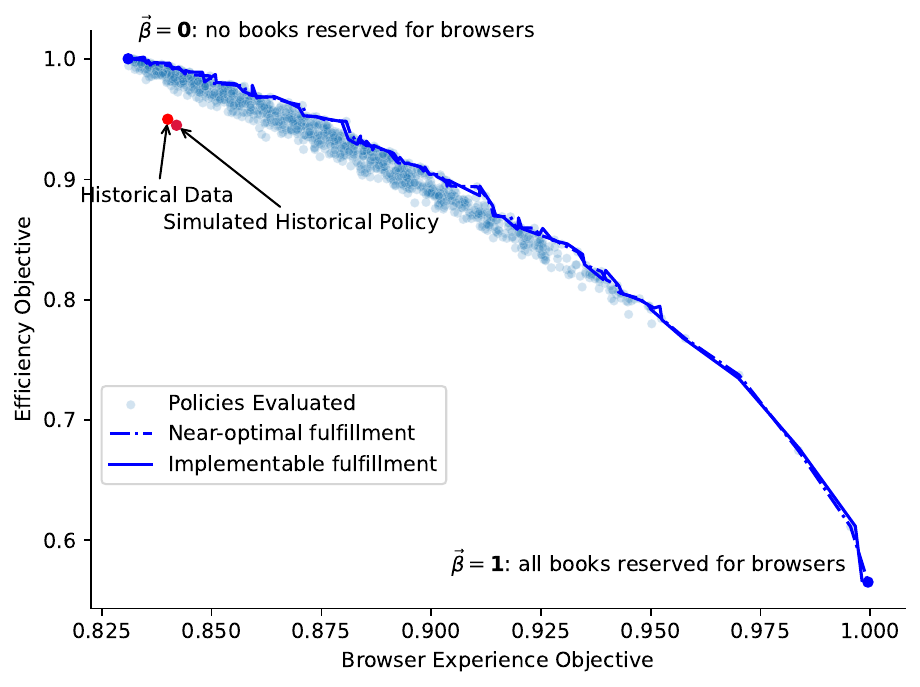}
    \caption{Pareto frontier found when browser experience is defined as the Nash welfare function of individual branches.}
    \label{fig:pareto_nash_app}
\end{figure}

\subsection{Relative contribution of fulfillment and inventory reserves}

A natural question arises from our analyses: what is the relative contribution of each policy lever? In \Cref{fig:relative_impact}, we present some evidence to answer this question. By adding near-optimal fulfillment to historical inventory reserve (i.e., no inventory reserve), we increase the usage objective by 0.050, and decrease the browser experience objective by 0.013. By subtracting browser reserve from the balanced implementable policy, we increase the usage objective by 0.033, and decrease the browser experience objective by 0.029. In another word, the two levers work in opposite directions, but when combined, they can get improvements in both objectives. One other comparison worth noting is between historical data and the balanced policy with no reserve: the only difference between these two settings is the fulfillment policy. We see that the latter Pareto dominates the former, suggesting historical fulfillment is sub-optimal in both objectives and that selecting better fulfillment, even without browser reserves, helps both objectives.

\begin{figure}[tbh]
    \centering
    \includegraphics[width=0.6\linewidth]{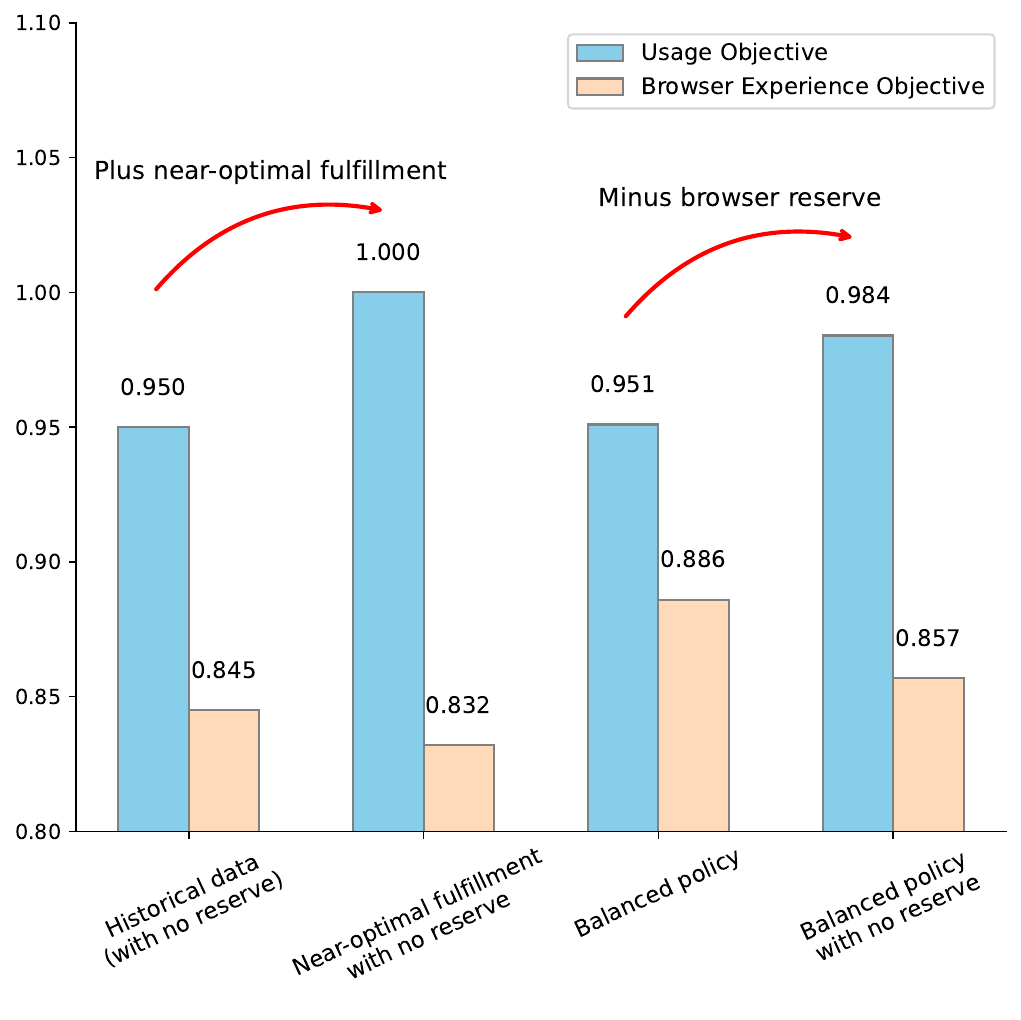}
    \caption{Comparison of usage and browser experience objectives for four policy scenarios. Historical data: calculated from historical data, note that historically there were (almost) no reserves for browsers (except for a few books at a select few branches, on special days). Near-optimal fulfillment with no reserve: using no browser reserve (same as historical), and the near-usage-optimal fulfillment method. Balanced policy: the policy presented in the main text. Balanced policy with no reserve: we take the implementable fulfillment of the balanced policy (i.e., \Cref{fig:implementable_fulfillment}), and set the browser reserves to 0. The first two settings represent the impact of near-optimal fulfillment; the last two settings represent the impact of browser reserves on the balanced policy. Furthermore, between historical data (first setting) and the balanced policy with no reserve (fourth setting), the only difference between them is the fulfillment policy. Selecting better fulfillment, even without browser reserves, helps both objectives.}
    \label{fig:relative_impact}
\end{figure}

\end{document}